\newcommand{\lv}[1]{#1}
\newcommand{\sv}[1]{}
\newcommand{\lsv}[2]{\lv{#1}\sv{#2}}
\documentclass{llncs}

\usepackage{times}

\usepackage{tikz}
\usetikzlibrary{positioning,calc}

\usepackage{url}
\usepackage{amsmath,amssymb}
\usepackage{xspace,tikz,enumitem}
\usepackage{mathrsfs,verbatim}
\usepackage{enumitem,todonotes,microtype}

\usepackage{color}
\usepackage{boxedminipage}

\usepackage[]{algorithm2e}

\newcommand{\Nat}{\mathbb{N}}
\newcommand{\hy}{\hbox{-}\nobreak\hskip0pt}

\newcommand{\bigoh}{\mathcal{O}}

\def\hy{\hbox{-}\nobreak\hskip0pt} 

\newcommand{\SB}{\{\,} \newcommand{\SM}{\;{|}\;} \newcommand{\SE}{\,\}}
 
\newcommand{\Card}[1]{|#1|}

\newcommand{\CC}{\mathcal{C}} 
 
\newcommand{\FF}{\mathcal{F}} 
 
\newcommand{\RR}{\mathcal{R}} 

\def\mx#1{\mbox{\boldmath$#1$}}

\def\MS#1{\mbox{MSO}}

\newcommand{\MSO}{\mbox{MSO}\xspace}

\def\HH{{\mathcal H}}

\def\PP{{\mathcal P}}

\spnewtheorem{ourfact}{Fact}{\bfseries}{\itshape}
\spnewtheorem{observation}{Observation}{\bfseries}{\itshape}

\def\sm{split-module}

\def\GLT{graph-labeled tree}
\def\wsm{well-structured modulator}

\newcommand{\MSOMC}[1]{\textsc{MSO-MC${}_{#1}$}}
\newcommand{\MSOOPT}[2]{\textsc{MSO-Opt${}^{#1}_{#2}$}}

\newcommand{\wsn}{\text{\normalfont\slshape wsn}}
\newcommand{\rw}{\text{\normalfont\slshape rw}}
\newcommand{\md}{\text{\normalfont\slshape mod}}

\begin{document}

\title{Solving Problems on Graphs of High Rank-Width\thanks{Supported by the Austrian Science Fund (FWF), project P26696.}}

\author{Eduard Eiben \and Robert Ganian \and Stefan Szeider}

\institute{Algorithms and Complexity Group, Institute of Computer Graphics and Algorithms\\ TU Wien,
  Vienna, Austria}

\maketitle

\begin{abstract}
\noindent A modulator of a graph $G$ to a specified graph class $\HH$ is a set of vertices whose deletion puts $G$ into $\HH$. The cardinality of a modulator to various graph classes has long been used as a structural parameter which can be exploited to obtain FPT algorithms for a range of hard problems.
  Here we investigate what happens when a graph contains a modulator which is large but ``well-structured'' (in the sense of having bounded rank-width). Can such modulators still be exploited to obtain efficient algorithms? And is it even possible to find such modulators efficiently?
  
  We first show that the parameters derived from such well-structured modulators are strictly more general than the cardinality of modulators and rank-width itself. Then, we develop an FPT algorithm for finding such well-structured modulators to any graph class which can be characterized by a finite set of forbidden induced subgraphs. We proceed by showing how well-structured modulators can be used to obtain efficient parameterized algorithms for \textsc{Minimum Vertex Cover} and \textsc{Maximum Clique}. Finally, we use the concept of well-structured modulators to develop an algorithmic meta-theorem for efficiently deciding problems expressible in Monadic Second Order (MSO) logic, and prove that this result is tight in the sense that it cannot be generalized to LinEMSO problems.
\end{abstract}


\section{Introduction}

Many important graph problems are known to be NP-hard, and yet admit efficient solutions in practice due to the inherent structure of instances. The parameterized complexity paradigm~\cite{DowneyFellows13,Niedermeier06} allows a more refined analysis of the complexity of various problems and hence enables the design of more efficient algorithms. In particular, given an instance of size $n$ and a numerical parameter $k$ which captures some property of the instance, one asks whether the instance can be solved in time $f(k)\cdot n^{\bigoh(1)}$. Parameterized problems which admit such an algorithm are called \emph{fixed parameter tractable} (FPT), and the algorithms themselves are often called FPT \emph{algorithms}.

Given the above, it is natural to ask what kind of structure can be exploited to obtain FPT algorithms for a wide range of natural graph problems. There are two very successful, mutually incomparable approaches which tackle this question.

\begin{enumerate}[leftmargin=*]
\item[A.] \emph{Width measures.} Treewidth has become an extremely successful structural parameter with a wide range of applications in many fields of computer science. However, treewidth is not suitable for use in dense graphs. This led to the development of algorithms that use the parameter clique-width~\cite{CourcelleMakowskyRotics00}, which can be viewed as a relaxation of treewidth towards dense graphs. However, while there are efficient theoretical algorithms for computing tree-decompositions, this is not the case for decompositions for clique-width. This shortcoming has later been overcome by the notion of rank-width~\cite{OumSeymour06}, which improves upon clique-width by allowing the efficient computation of rank-decompositions while retaining all of the positive algorithmic results previously obtained for clique-width.
\item[B.] \sloppypar \emph{Modulators.} A modulator is a vertex set whose deletion places the considered graph into some specified graph class. A substantial amount of research has been placed into finding as well as exploiting small modulators to various graph classes~\cite{GajarskyHlinenyObdrzalek13,BodlaenderJansenKratsch13}. Popular notions such as vertex cover and feedback vertex set are also special cases of modulators (to the classes of edgeless graphs and forests, respectively). One advantage of parameterizing by the size of modulators is that it allows us to build on the vast array of research of polynomial-time algorithms on specific graph classes (see, for instance,~\cite{CorneilLerchsBurlingham81,LokshtanovVatshelleVillanger14}). In other fields of computer science, modulators are often called \emph{backdoors} and have been successfully used to obtain efficient algorithms for, e.g., Satisfiability 
and Constraint Satisfaction~\cite{GaspersMisraOSZ14}.
\end{enumerate}
Our primary goal in this paper is to push the boundaries of tractability for a wide range of problems above the state of the art for both of these approaches. We summarize our contributions below.\smallskip

\begin{enumerate}[leftmargin=*,nosep]
\item We introduce a family of ``hybrid'' parameters that combine approaches A and B. 
\end{enumerate}
\smallskip 
Given a graph $G$ and a fixed graph class $\HH$, the new parameters capture (roughly speaking) the minimum rank-width of any modulator of $G$ into $\HH$. We call this the \emph{well-structure number} of $G$ or $\wsn^\HH(G)$. The formal definition of the parameter also relies on the notion of \emph{split decompositions}~\cite{Cunningham82} and is provided in Section~\ref{sec:wsm}, where we also prove that for any graph class $\HH$ of unbounded rank-width, $\wsn^\HH$ is not larger and in many cases much smaller than both rank-width and the size of a modulator to $\HH$. 
\smallskip

\begin{enumerate}[leftmargin=*,nosep]
\item[2.] We develop an FPT algorithm for computing $\wsn^\HH$.
\end{enumerate}
\smallskip
As with most structural parameters, virtually all algorithmic applications of the well-structure number rely on having access to an appropriate decomposition. In Section~\ref{sec:finding} we provide an FPT algorithm for computing $\wsn^\HH$ along with the corresponding decomposition for any graph class $\HH$ which can be characterized by a finite set of forbidden induced subgraphs (\emph{obstructions}). This is achieved by building on the polynomial algorithm for computing split-decompositions~\cite{GioanPaulTedderCorneil14} in combination with the FPT algorithm for computing rank-width~\cite{HlinenyOum08}.
\smallskip

\begin{enumerate}[leftmargin=*,nosep]
\item[3.] We design FPT algorithms for Minimum Vertex Cover (\textsc{MinVC}) and Maximum Clique (\textsc{MaxClq}) parameterized by $\wsn^\HH$.
\end{enumerate}
\smallskip
Specifically, in Section~\ref{sec:examples} we show that for any graph class $\HH$ (which can be characterized by a finite set of obstructions) such that the problem is polynomial-time tractable on $\HH$, the problem becomes fixed parameter tractable when parameterized by $\wsn^\HH$. We also give an overview of possible choices of $\HH$ for \textsc{MinVC} and \textsc{MaxClq}.
\smallskip

\begin{enumerate}[leftmargin=*,nosep]
\item[4.] We develop a \emph{meta-theorem} to obtain FPT algorithms for problems definable in Monadic Second Order (MSO) logic~\cite{CourcelleMakowskyRotics00} parameterized by $\wsn^\HH$.
\end{enumerate}
\smallskip
The meta-theorem requires that the problem is FPT when parameterized by the cardinality of a modulator to $\HH$. We prove that this condition is not only sufficient but also necessary, in the sense that the weaker condition of polynomial-time tractability on $\HH$ used for \textsc{MinVC} and \textsc{MaxClq} is not sufficient for FPT-time MSO model checking. Formal statements and proofs can be found in Section~\ref{sec:mso}.
\smallskip

\begin{enumerate}[leftmargin=*,nosep]
\item[5.] We show that, in general, solving LinEMSO problems~\cite{CourcelleMakowskyRotics00,GanianHlineny10} is not FPT when parameterized by $\wsn^\HH$.
\end{enumerate}
\smallskip
In particular, in the concluding Section~\ref{sec:hardness} we give a proof that these problems are in general paraNP-hard when parameterized by $\wsn^\HH$ under the same conditions as those used for MSO model checking. 



\section{Preliminaries}\label{sec:prel}
The set of natural numbers (that is, positive integers) will be denoted by
$\Nat$. For $i \in \Nat$ we write $[i]$ to denote the set $\{1,
\dots, i \}$. If $\sim$ is an equivalence relation over a set $A$, then for $a\in A$ we use $[a]_{\sim}$ to denote the equivalence class containing $a$.

\paragraph{Graphs} We will use standard graph theoretic terminology and notation
(cf. \cite{Diestel00}). All graphs considered in this document are simple and undirected. \lv{The non-leaf vertices of a tree are called its \emph{internal nodes}. 
If $S$ is a set of leaves of $T$ , then $T (S)$ denotes
the smallest connected subtree spanning $S$. }

Given a graph $G=(V(G),E(G))$ and $A\subseteq V(G)$, we denote by $N(A)$ the set of neighbors of $A$ in $V(G)\setminus A$; if $A$ contains a single vertex $v$, we use $N(v)$ instead of $N(\{v\})$. We use $V$ and $E$ as shorthand for $V(G)$ and $E(G)$, respectively, when the graph is clear from context. Two vertex sets $A,B$ are \emph{overlapping} if $A\cap B, A\setminus B, B\setminus A$ are all nonempty. $G-A$ denotes the subgraph of $G$ obtained by deleting $A$.

Given a graph $G=(V,E)$ and a graph class $\HH$, a set $X\subseteq V$ is called a \emph{modulator} to $\HH$ if $G-X\in \HH$. A graph class is called \emph{hereditary} if it is closed under vertex deletion. A graph $H$ is an \emph{induced subgraph} of $G$ if $H$ can be obtained by deleting vertices (along with all of their incident edges) from $G$. For $A\subseteq V(G)$ we use $G[A]$ to denote the subgraph of $G$ obtained by deleting $V(G)\setminus A$. Let $\FF$ be a finite set of graphs; then the class of $\FF$-\emph{free} graphs is the class of all graphs which do not contain any graph in $\FF$ as an induced subgraph. We will often refer to elements of $\FF$ as \emph{obstructions}, and we say that the class of $\FF$-free graphs is \emph{characterized by $\FF$}.

\paragraph{Fixed-Parameter Tractability.}

We refer the reader to~\cite{DowneyFellows13,Niedermeier06} for an introduction to parameterized complexity.
A \emph{parameterized problem} $\mathcal{P}$ is a subset of $\Sigma^* \times
\Nat$ for some finite alphabet $\Sigma$. For a problem instance $(x,k)
\in \Sigma^* \times \Nat$ we call $x$ the main part and $k$ the
parameter.  A parameterized problem $\mathcal{P}$ is \emph{fixed-parameter
  tractable} (FPT in short) if a given instance $(x, k)$ can be solved in time
$O(f(k) \cdot p(\Card{x}))$ where $f$ is an arbitrary computable
function of $k$ and $p$ is a polynomial function.

\lv{\paragraph{Splits and Graph Labeled Trees}}
\sv{\paragraph{Splits.}}
  A \emph{split} of a connected graph $G=(V,E)$ is a vertex bipartition $\{A,B\}$ of $V$
 such that every vertex of $A' = N(B)$ has the same neighborhood in $B'=N(A)$. The sets $A'$ and $B'$ are called \emph{frontiers} of the split. 
 \newcommand{\prelimsplitsa}[0]{
 A split is said to be \emph{non-trivial} if both sides have at least two vertices. A connected graph which does not contain a non-trivial split is called \emph{prime}. A bipartition is \emph{trivial} if one of its parts is the empty set or a singleton. Cliques and stars are called \emph{degenerate} graphs; notice that every non-trivial bipartition of their vertices is a split.}
 \lv{\prelimsplitsa}

Let $G=(V,E)$ be a graph. To simplify our exposition, we will use the notion of \emph{\sm s} instead of splits where suitable. A set $A\subseteq V$ is called a \emph{\sm} of $G$ if there exists a connected component $G'=(V',E')$ of $G$ such that $\{A,V'\setminus A\}$ forms a split of $G'$. Notice that if $A$ is a {\sm} then $A$ can be partitioned into $A_1$ and $A_2$ such that $N(A_2)\subseteq A$ and for each $v_1, v_2\in A_1$ it holds that $N(v_1)\cap (V'\setminus A)=N(v_2)\cap (V'\setminus A)$. For technical reasons, $V$ and $\emptyset$ are also considered \sm s. We say that two disjoint \sm s~$X, Y \subseteq V$ are \emph{adjacent} if there exist
$x \in X$ and $y \in Y$ such that $x$ and $y$ are adjacent.

 \newcommand{\prelimsplitsb}[0]{
A \emph{graph-labeled tree} is a pair $(T, \FF)$, where $T$ is a tree and $\FF$ is a set of graphs such that each internal node $u$ of $T$ is \emph{labeled} by a graph $G(u)\in \FF$ and there is a bijection between the edges of $T$ incident to $u$ and vertices of $G(u)$. When clear from the context, we may use $u$ as a shorthard for $G(u) \in \FF$; for instance, 
we use $V(u)$ to denote $V(G(u))$ and we say that an edge of $T$ incident to $u$ is \emph{incident} to the vertex of $G(u)$ mapped to it. Graph-labeled trees were introduced in \cite{GioanPaul07,GioanPaul12} and in the following paragraphs we recall some useful definitions and theorems that appear in \cite{GioanPaulTedderCorneil14}.

For an internal node $u$ of $T$, the vertices of $V(u)$ are called \emph{marker} vertices and the edges of $E(u)$ are called \emph{label-edges}. Edges of $T$ incident to two internal nodes are called \emph{tree-edges}. Marker vertices incident to a tree-edge $e$ are called the \emph{extremities} of $e$, and each leaf $v$ is \emph{associated with} the unique marker vertex $q$ (in the neighbor of $v$) mapped to the edge incident to $v$.
Perhaps the most important notion for graph-labeled trees with respect to split decomposition is that of \emph{accessibility}.

\begin{definition}
Let $(T, \FF)$ be a graph-labeled tree. The marker vertices $q$ and $q'$ are accessible from one another if there is a sequence $\Pi$ of marker vertices $q,\dots,q'$ such that the two following conditions holds. 
\begin{enumerate}
\item Every two consecutive elements of~$\Pi$ are either the vertices of a label-edge or the
extremities of a tree-edge;
\item the sequence of edges obtained above alternates between tree-edges and label-edges.
\end{enumerate}
\end{definition} 

\begin{figure}
\begin{center}
\scalebox{0.8}{
\begin{tikzpicture}[every node/.style={circle, inner sep=0.04cm, fill=black, draw, scale=0.9}, scale=1.0, rotate = 180, xscale = -1]

\node[label=below:{1}] (1) at ( 1, 4) {};
\node[label=right:{2}] (2) at ( 1, 3) {};
\node[label=above:{3}] (3) at ( 3, 2) {};
\node[label=above:{4}] (4) at ( 5., 3) {};
\node[label=above:{5}] (5) at ( 6, 3) {};
\node[label=below:{6}] (6) at ( 6, 4) {};
\node[label=below:{7}] (7) at ( 5, 4) {};
\node[label=below:{8}] (14) at ( 2.5, 5.3) {};
\node[label=below:{9}] (15) at ( 3.5, 5.3) {};

\draw (1) -- (2);
\draw (1) -- (3);
\draw (4) -- (3);
\draw (5) -- (4);
\draw (6) -- (5);
\draw (7) -- (6);
\draw (15) -- (7);
\draw (15) -- (4);
\draw (14) -- (7);
\draw (14) -- (4);
\draw (15) -- (3);
\draw (14) -- (3);
\draw (14) -- (1);
\draw (15) -- (1);
\draw (7) -- (1);
\draw (4) -- (1);
\draw (14) -- (15);
    \tikzstyle{empty}=[draw, shape=circle, minimum size=3pt,inner sep=0pt, fill=white,color=white]
\node[empty] () [below = 1.8cm] at (1) {};
\end{tikzpicture}
}
\quad
\scalebox{0.5}{
\begin{tikzpicture}[remember picture, 
outer/.style={ scale=1.5},
inner/.style= {circle, fill=black, draw, scale=.5},
scale=1.0, rotate = 180, xscale = -1]

		\node[inner] (a1) {};
		\node[inner, above right=of a1] (a2) {};
		\node[inner, below right=of a2] (a3) {};
        \node[below left=of a3] (a4) {};

		\draw (a1) -- (a2);
		\draw (a1) -- (a3);

		\node[inner, right=2cm of a3] (b1) {};
		\node[inner, above right=of b1] (b2) {};
		\node[inner, below right=of b2] (b3) {};
        \node[inner, below left=of b3] (b4) {};

		\draw (b1) -- (b2);
		\draw (b1) -- (b3);
		\draw (b1) -- (b4);
		\draw (b2) -- (b3);
		\draw (b2) -- (b4);
		\draw (b3) -- (b4);

		\node[inner, right= 2cm of b3] (c1) {};
		\node[inner, above right=of c1] (c2) {};
		\node[inner, below right=of c1] (c3) {};
        \node[inner, right=of c2] (c4) {};
        \node[inner, right=of c3] (c5) {};
        
        \draw (c1) -- (c2) -- (c4) -- (c5) -- (c3) -- (c1);

		\node[inner, below= 2cm of b4] (d1) {};
		\node[inner, below left=of d1] (d2) {};
		\node[inner, below right=of d1] (d3) {};

		\draw (d1) -- (d2);
		\draw (d1) -- (d3);
        \draw (d2) -- (d3);

\node[outer, left=1cm of a1] (1) {1};
\node[outer, above=0.7cm of a2] (2) {2};
\node[outer, above=0.7cm of b2] (3) {3};
\node[outer, above=1cm of c2] (4) {4};
\node[outer, above right=0.7cm of c4] (5) {5};
\node[outer, below=1cm of c3] (6) {6};
\node[outer, below right=0.7cm of c5] (7) {7};
\node[outer, below left=0.4cm and 0.8cm of d2] (8) {8};
\node[outer, below right=0.4cm and 0.8cm of d3] (9) {9};

\draw (a3) -- (b1);
\draw (b3) -- (c1);
\draw (b4) -- (d1);

 \draw (a1) -- (1);
 \draw (a2) -- (2);
 \draw (b2) -- (3);
 \draw (c2) -- (4);
 \draw (c4) -- (5);
 \draw (c3) -- (6);
 \draw (c5) -- (7);
 \draw (d2) -- (8);
 \draw (d3) -- (9);

\draw (1.13,0) circle (1.5cm);
\draw (5.53,0) circle (1.5cm);
\draw (10.23,0) circle (1.7cm);
\draw (5.53,4.23) circle (1.5cm);

\end{tikzpicture}
}
\end{center}
\caption{A graph-labeled tree (right) and its accessibility graph (left).}
\end{figure}
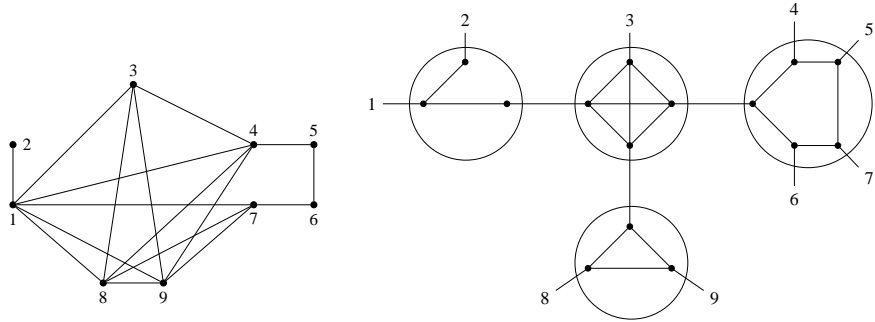

Two leaves are accessible if their associated marker vertices are accessible. The \emph{accessibility graph} of graph-labeled tree $(T,\FF)$, denoted $Gr(T,\FF)$, is the graph whose vertices are leaves of $T$ and which has an edge between two distinct leaves $l$ and $l'$ if and only if they are accessible from one another. Conversely, we may say that $(T,\FF)$ is the graph-labeled tree of $Gr(T,\FF)$.

\begin{definition}[\cite{GioanPaulTedderCorneil14}]
Let $e$ be a tree-edge incident to internal nodes $u$ and $u'$ in a \GLT, and let
$q\in V (u)$ and $q' \in V (u')$ be the extremities of $e$. The \emph{node-join} of $u$, $u'$ replaces $u$
and $u'$ with a new internal node $v$ labeled by the graph formed from the disjoint union of $G(u)$ and $G(u')$ as
follows: all possible label-edges are added between $N (q)$ and $N (q' )$, and then $q$ and
$q'$ are deleted. The new node $v$ is made adjacent to all neighbors of $u$ and $u'$ in $T$.
The \emph{node-split} is then the inverse of the node-join.
\end{definition}

Notice that the node-join operation and the node-split operation preserve the
accessibility graph of the GLT. A graph-labeled tree is \emph{reduced} if all its labels are either prime or degenerate, and no node-join of two cliques or two stars is possible.

\begin{theorem}[\cite{Cunningham82,GioanPaul07,GioanPaul12,GioanPaulTedderCorneil14}]\label{thm:uniquereducedGLT}
For any connected graph $G$, there exists a unique, reduced graph-labeled tree $(T , \FF)$ such that $G = Gr(T , \FF)$.
\end{theorem}

The unique \GLT~guaranteed by the previous theorem is the \emph{split-tree}, and is denoted $ST(G)$.

\begin{theorem}[\cite{Cunningham82,GioanPaul07,GioanPaul12,GioanPaulTedderCorneil14}]\label{thm:splitinGLT}
Let   $(T , \FF)$ be the split-tree of a connected graph $G$. Any split of $G$ is the bipartition
(of leaves) induced by removing an internal tree-edge from $T'$, where $T' = T$ or $T'$ is
obtained from $T$ by exactly one node-split of a degenerate node.
\end{theorem}

\begin{theorem}[\cite{GioanPaulTedderCorneil14}]\label{thm:computingGLT}
The split-tree $ST(G)$ of a connected graph $G = (V , E)$ with n vertices and
m edges can be built incrementally in time $O(n +m)\alpha(n + m)$, where $\alpha$ is the inverse Ackermann function.
\end{theorem}
}

 \lv{\prelimsplitsb}

\paragraph{Rank-width}

For a graph $G$ and $U,W\subseteq V(G)$, let $\mx A_G[U,W]$ denote the
$U\times W$-submatrix of the adjacency matrix over the two-element
field $\mathrm{GF}(2)$, i.e., the entry $a_{u,w}$, $u\in U$ and $w\in
W$, of $\mx A_G[U,W]$ is $1$ if and only if $\{u,w\}$ is an edge
of~$G$.  The {\em cut-rank} function $\rho_G$ of a graph $G$ is
defined as follows: For a bipartition $(U,W)$ of the vertex
set~$V(G)$, $\rho_G(U)=\rho_G(W)$ equals the rank of $\mx A_G[U,W]$
over $\mathrm{GF}(2)$. 

A \emph{rank-decomposition} of a graph $G$ is a pair $(T,\mu)$
where $T$ is a tree of maximum degree 3 and $\mu:V(G)\rightarrow \{t:
\text{$t$ is a leaf of $T$}\}$ is a bijective function. For an edge~$e$ of~$T$, the connected components of $T - e$ induce a
bipartition $(X,Y)$ of the set of leaves of~$T$.  The \emph{width} of
an edge $e$ of a rank-decomposition $(T,\mu)$ is $\rho_G(\mu^{-1} (X))$.
The \emph{width} of $(T,\mu)$ is the maximum width over all edges of~$T$.  The \emph{rank-width} of $G$, $\rw(G)$ in short, is the minimum width over all
rank-decompositions of $G$. We denote by $\RR_i$ the class of all graphs of rank-width at most $i$, and say that a graph class $\HH$ is of \emph{unbounded rank-width} if $\HH\not \subseteq \RR_i$ for any $i\in \Nat$.


\newenvironment{psmallmatrix}
  {\left(\begin{smallmatrix}}
  {\end{smallmatrix}\right)}

\begin{figure}[ht]
\vspace{-0.6cm}
  \centering
  \tikzstyle{every circle node}=[circle,draw,inner sep=1.0pt, fill=black]
  \begin{tikzpicture}
    \begin{scope}[xshift=-7cm]
      \draw
      (90+4*72:1) node[circle] (d)  {} node[right=2pt] {$\strut d$} 
      (90+0*72:1) node[circle] (c)  {} node[above] {$\strut c$} 
      (90+1*72:1) node[circle] (b)  {} node[left=2pt] {$\strut b$} 
      (90+2*72:1) node[circle] (a)  {} node[below] {$\strut a$} 
      (90+3*72:1) node[circle] (e)  {} node[below] {$\strut e$} 
      (a)--(b)--(c)--(d)--(e)--(a)
      ;
    \end{scope}

 \draw 
 (-3,1)   node[circle] (d) {} node[above] {$\strut d$} 
 (-3,-1)  node[circle] (e) {} node[below] {$\strut e$} 
 (-2,0)   node[circle] (U) {} 
 (0,0)    node[circle] (V) {} 
 (0,-1.2)    node[circle] (X) {} node[below] {$\strut a$} 
 (2,0)    node[circle] (W) {}
 (3,1)   node[circle] (c) {} node[above] {$\strut c$} 
 (3,-1)  node[circle] (b) {} node[below] {$\strut b$} 

 (d) to node [auto,near start, swap] {$\begin{psmallmatrix}
                              0&0&1&1
                            \end{psmallmatrix}$} (U)
 (e) to node [auto, near start] {$\begin{psmallmatrix}
                              1&0&0&1
                            \end{psmallmatrix}$} (U)
 (U) to node [auto] {$\begin{psmallmatrix}
                              0&0&1\\
                              1&0&0  
                            \end{psmallmatrix}$} (V)
 (V) to node [auto] {$\begin{psmallmatrix}
                              1&0\\
                              0&1\\
                              0&0
                            \end{psmallmatrix}$} (W)
 (V)  to node [auto] {$\begin{psmallmatrix}
                              1\\
                              0\\
                              0\\
                              1
                            \end{psmallmatrix}$} (X)
 (c)--(W)--(b)
(3.4,0.5) node {$\begin{psmallmatrix}
                              0\\
                              1\\
                              1\\
                              0
                            \end{psmallmatrix}$} 
(3.4,-0.5) node {$\begin{psmallmatrix}
                              1\\
                              1\\
                              0\\
                              0
                            \end{psmallmatrix}$} 
;
\end{tikzpicture}
\vspace{-0.8cm}
\caption{A rank-decomposition of the cycle $C_5$.}
\label{fig:rdecC5}
\vspace{-0.6cm}
\end{figure}
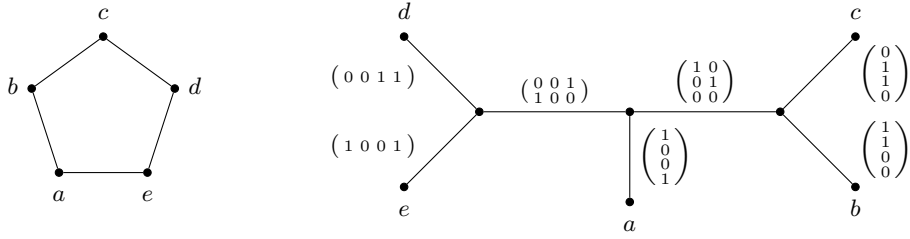

\begin{theorem}[\cite{HlinenyOum08}]\label{thm:rankdecomp} Let $k \in \Nat$ be a constant and
  $n \geq 2$. For an $n$-vertex graph $G$, we can output a
  rank-decomposition of width at most $k$ or confirm that the
  rank-width of $G$ is larger than $k$ in time $f(k)\cdot n^3$, where $f$ is a computable function.
\end{theorem}

\paragraph{Monadic Second Order Logic on Graphs}
We assume that we have an infinite supply of individual variables,
denoted by lowercase letters $x,y,z$, and an infinite supply of set
variables, denoted by uppercase letters $X,Y,Z$. \emph{Formulas} of
\emph{monadic second-order logic} (MSO) are constructed from atomic
formulas $E(x,y)$, $X(x)$, and $x = y$ using the connectives $\neg$
(negation), $\wedge$ (conjunction) and existential quantification
$\exists x$ over individual variables as well as existential
quantification $\exists X$ over set variables. Individual variables
range over vertices, and set variables range over sets of
vertices. The atomic formula $E(x,y)$ expresses adjacency, $x = y$
expresses equality, and $X(x)$ expresses that vertex $x$ in the set
$X$. From this, we define the semantics of monadic second-order logic
in the standard way (this logic is sometimes called $\MSO_1$).

\emph{Free and bound variables} of a formula are defined in the usual way. A
\emph{sentence} is a formula without free variables. We write $\varphi(X_1,
\dots, X_n)$ to indicate that the set of free variables of formula $\varphi$
is $\{X_1, \dots, X_n\}$. If $G = (V,E)$ is a graph and $S_1, \dots, S_n
\subseteq V$ we write $G \models \varphi(S_1, \dots, S_n)$ to denote that
$\varphi$ holds in $G$ if the variables $X_i$ are interpreted by the sets
$S_i$, for $i \in [n]$. For a fixed \MSO sentence $\varphi$, the MSO Model Checking problem ($\MSOMC{\varphi}$) asks whether an input graph $G$ satisfies $G\models \varphi$.

It is known that MSO formulas can be checked efficiently as long as the graph has bounded rank-width.

\begin{theorem}[\cite{GanianHlineny10}]\label{thm:msorankwidth}
  Let $\varphi$ and $\psi=\psi(X)$ be fixed
  \MSO formulas. Given an $n$-vertex graph $G$ and a set $S\subseteq V(G)$, there exists a computable function $f$ such that we can decide whether $G \models \varphi$ and whether $G\models \psi(S)$ in time $f(\rw(G))\cdot n^3$.
\end{theorem}

We review \MSO \emph{types} roughly following the presentation in
\cite{Libkin04}. The \emph{quantifier rank} of an \MSO formula $\varphi$ is
defined as the nesting depth of quantifiers in $\varphi$. For non\hy negative
integers $q$ and $l$, let $\MSO_{q,l}$ consist of all \MSO formulas of
quantifier rank at most $q$ with free set variables in $\{X_1, \dots, X_l\}$.

Let $\varphi = \varphi(X_1,\dots,X_l)$ and $\psi = \psi(X_1,\dots,X_l)$ be
\MSO formulas. We say $\varphi$ and $\psi$ are \emph{equivalent}, written $\varphi
\equiv \psi$, if for all graphs $G$ and $U_1, \dots, U_l \subseteq V(G)$, $G
\models \varphi(U_1,\dots, U_l)$ if and only if $G \models \psi(U_1,\dots,
U_l)$.  Given a set $F$ of formulas, let ${F/\mathord\equiv}$ denote the set
of equivalence classes of $F$ with respect to $\equiv$. A system of
representatives of $F/\mathord\equiv$ is a set $R \subseteq F$ such that $R
\cap C \neq \emptyset$ for each equivalence class $C \in F/\mathord\equiv$.
The following statement has a straightforward proof using normal forms (see
\cite[Proposition~7.5]{Libkin04} for details).
\begin{ourfact}\label{fact:representatives}
  Let $q$ and $l$ be fixed non\hy negative integers. The set
  $\MSO_{q,l}/\mathord\equiv$ is finite, and one can compute a system of
  representatives of $\MSO_{q,l}/\mathord\equiv$.
\end{ourfact}
We will assume that for any pair of non\hy negative integers $q$ and $l$ the
system of representatives of $\MSO_{q,l}/\mathord\equiv$ given by
Fact~\ref{fact:representatives} is fixed.
\begin{definition}[\MSO Type]
  Let $q,l$ be non\hy negative integers. For a graph $G$ and an $l$\hy
  tuple $\vec{U}$ of sets of vertices of $G$, we define
  $\mathit{type}_q(G,\vec{U})$ as the set of formulas $\varphi \in
  \MSO_{q,l}$ such that $G \models \varphi(\vec{U})$. We call
  $\mathit{type}_q(G,\vec{U})$ the \MSO \emph{$q$-type of
    $\vec{U}$ in $G$}. 
\end{definition}
It follows from Fact~\ref{fact:representatives} that up to logical
equivalence, every type contains only finitely many formulas. \lv{This
allows us to represent types using \MSO formulas as follows.}

\newcommand{\lemtypeformula}[0]{
\begin{lemma}[\cite{GanianSlivovskySzeider13}]
\label{lem:typeformula}
  Let $q$ and $l$ be non\hy negative integer constants, let $G$ be a graph,
  and let $\vec{U}$ be an $l$\hy tuple of sets of vertices of $G$. One can
  compute a formula $\Phi \in \MSO_{q,l}$ such that for any graph
  $G'$ and any $l$\hy tuple $\vec{U}'$ of sets of vertices of $G'$ we have $G'
  \models \Phi(\vec{U}')$ if and only if $\mathit{type}_q(G,\vec{U}) =
  \mathit{type}_q(G',\vec{U}')$. Moreover, $\Phi$ can be computed in time $\bigoh(f(\rw(G))\cdot |V|^{\bigoh(1)})$.
\end{lemma}

\begin{proof}
  Let $R$ be a system of representatives of $\MSO_{q,l}/\mathord\equiv$ given
  by Fact~\ref{fact:representatives}. Because $q$ and $l$ are constant, we can
  consider both the cardinality of $R$ and the time required to compute it as
  constants. Let $\Phi \in \MSO_{q,l}$ be the formula defined as $\Phi =
  \bigwedge_{\varphi \in S} \varphi \wedge \bigwedge_{\varphi \in R \setminus
    S} \neg \varphi$, where $S = \SB \varphi \in R \SM G \models
  \varphi(\vec{U}) \SE$. We can compute $\Phi$ by deciding $G \models
  \varphi(\vec{U})$ for each $\varphi \in R$. Since the number of formulas in
  $R$ is a constant, this can be done in time $\bigoh(f(\rw(G))\cdot |V|^{\bigoh(1)})$ if $G \models \varphi(\vec{U})$ can
  be decided in time $f(\rw(G))\cdot |V|^{\bigoh(1)}$.

  Let $G'$ be an arbitrary graph and let $\vec{U}'$ be an $l$\hy tuple of subsets of
  $V(G')$. We claim that $\mathit{type}_q(G, \vec{U}) = \mathit{type}_q(G',
  \vec{U'})$ if and only if $G' \models \Phi(\vec{U}')$. Since $\Phi \in
  \MSO_{q,l}$ the forward direction is trivial. For the converse, assume
  $\mathit{type}_q(G, \vec{U}) \neq \mathit{type}_q(G', \vec{U'})$. First
  suppose $\varphi \in \mathit{type}_q(G, \vec{U}) \setminus
  \mathit{type}_q(G', \vec{U'})$. The set $R$ is a system of representatives
  of $\MSO_{q,l}/\mathord\equiv$ , so there has to be a $\psi \in R$ such that
  $\psi \equiv \varphi$. But $G' \models \Phi(\vec{U}')$ implies $G' \models
  \psi(\vec{U}')$ by construction of $\Phi$ and thus $G' \models
  \varphi(\vec{U}')$, a contradiction. Now suppose $\varphi \in
  \mathit{type}_q(G', \vec{U}') \setminus \mathit{type}_q(G, \vec{U})$. An
  analogous argument proves that there has to be a $\psi \in R$ such that
  $\psi \equiv \varphi$ and $G' \models \neg \psi(\vec{U}')$. It follows that
  $G' \not \models \varphi(\vec{U}')$, which again yields a contradiction.
  \qed
\end{proof}}
\lv{\lemtypeformula}

\newcommand{\MSOgames}[0]{
\begin{definition}[Partial isomorphism]\label{def:partialisomorphism}
  Let $G, G'$ be graphs, and let $\vec{V} = (V_1, \dots, V_l)$ and
  $\vec{U} = (U_1, \dots, U_l)$ be tuples of sets of vertices with
  $V_i \subseteq V(G)$ and $U_i \subseteq V(G')$ for each $i \in
  [l]$. Let $\vec{v} = (v_1, \dots, v_m)$ and $\vec{u} = (u_1, \dots,
  u_m)$ be tuples of vertices with $v_i \in V(G)$ and $u_i \in V(G')$
  for each $i \in [m]$. Then $(\vec{v}, \vec{u})$ defines a
  \emph{partial isomorphism between $(G, \vec{V})$ and $(G',
  \vec{U})$} if the following conditions hold:
  \begin{itemize}
    \item For every $i,j \in [m]$,
    \begin{align*}
      v_i = v_j \: \Leftrightarrow \: u_i = u_j \text{ and }
      v_iv_j \in E(G)\: \Leftrightarrow \: u_iu_j \in E(G').
    \end{align*}
    \item For every $i \in [m]$ and $j \in [l]$,
      \begin{align*}
        v_i \in V_j \: \Leftrightarrow u_i \in U_j.
      \end{align*}
    \end{itemize}
\end{definition}

\begin{definition}
  Let $G$ and $G'$ be graphs, and let $\vec{V_0}$ be a $k$\hy tuple of subsets
  of $V(G)$ and let $\vec{U_0}$ be a $k$\hy tuple of subsets of $V(G')$. Let
  $q$ be a non\hy negative integer. The \emph{$q$\hy round \MSO game on $G$
    and $G'$ starting from $(\vec{V_0}, \vec{U_0})$} is played as follows.
  The game proceeds in rounds, and each round consists of one of the following
  kinds of moves.
\begin{itemize}
  \item \textbf{Point move} The Spoiler picks a vertex in either $G$ or $G'$; the Duplicator responds by picking a vertex in the other graph.
  \item \textbf{Set move} The Spoiler picks a subset of $V(G)$ or a
    subset of $V(G')$; the Duplicator responds by picking a subset of the
    vertex set of the other graph.
  \end{itemize}
  Let $\vec{v}=(v_1,\dots,v_m), v_i \in V(G)$ and $
\vec{u}=(u_1,\dots,u_m), u_i \in V(G')$ be the point
  moves played in the $q$-round game, and let $\vec{V}=(V_1, \dots, V_l), V_i\subseteq V(G)$ and $\vec{U}=(U_1,\dots, U_l), U_i \subseteq V(G')$ be the set moves played in the
  $q$\hy round game, so that $l + m = q$ and moves belonging to same round
  have the same index. Then the Duplicator wins the game if $(\vec{v},
  \vec{u})$ is a partial isomorphism of $(G, \vec{V_0}\cup\vec{V})$ and $(G',
  \vec{U_0}\cup \vec{U})$. If the Duplicator has a winning strategy, we write $(G,
  \vec{V_0}) \equiv^{\MSO}_q (G', \vec{U_0})$.
\end{definition}

\begin{theorem}[\cite{Libkin04}, Theorem 7.7]\label{thm:msogames} Given two graphs $G$ and $G'$ and two $l$\hy tuples $\vec{V_0}, \vec{U_0}$ of sets of vertices of $G$ and $G'$, we have \begin{align*}
    \mathit{type}_q(G, \vec{V_0}) = \mathit{type}_q(G, \vec{U_0}) \:
    \Leftrightarrow \: (G, \vec{V_0}) \equiv^{\MSO}_q (G', \vec{U_0}).
\end{align*}
\end{theorem}}

\lv{\MSOgames}

\section{Well-Structured Modulators}\label{sec:wsm}

\begin{definition}
\label{def:wsm}
Let $\HH$ be a hereditary graph class and let $G$ be a graph. A set $\vec{X}$ of pairwise-disjoint split-modules of $G$ is called a $k$-\emph{\wsm}{} to $\HH$ if
\begin{enumerate}
\item $|\vec{X}|\leq k$, and
\item $\bigcup_{X_i\in \vec{X}} X_i$ is a modulator to $\HH$, and
\item $\rw(G[X_i])\leq k$ for each $X_i\in \vec{X}$.
\end{enumerate}
\end{definition}

\begin{figure}
\centering
\begin{tikzpicture}[every node/.style={circle, fill=black, draw, scale=.3}, scale=0.5, rotate = 180, xscale = -0.8, yscale = 0.5]

\filldraw[fill opacity=0.7,fill=gray!20] (0,0) ellipse (3.7cm and 2.8cm);
\filldraw[fill opacity=0.7,fill=gray!20] (1.6,5) ellipse (2.2cm and 1.7cm);

 \foreach \x in {1,...,8}{%
   \pgfmathparse{(\x-1)*45+floor(\x/9)*45}
    \node (N-\x) at (\pgfmathresult:2.4cm) [thick] {};
 } 
 \foreach \x [count=\xi from 1] in {2,...,8}{%
    \foreach \y in {\x,...,8}{%
        \path (N-\xi) edge[-] (N-\y);
  }
 }

 \newlength{\gridsize}
\setlength{\gridsize}{0.3cm}

\node[below right=0.1cm and 1cm of N-1] (1)  {};
\node[right=\gridsize of 1] (2) {};
\node[right=\gridsize of 2] (3) {};
\node[right=\gridsize of 3] (4) {};
\node[below=\gridsize of 1] (5) {};
\node[below=\gridsize of 2] (6) {};
\node[below=\gridsize of 3] (7) {};
\node[below=\gridsize of 4] (8) {};
\node[below=\gridsize of 5] (9) {};
\node[below=\gridsize of 6] (10) {};
\node[below=\gridsize of 7] (11) {};
\node[below=\gridsize of 8] (12) {};
\node[below=\gridsize of 9] (13) {};
\node[below=\gridsize of 10] (14) {};
\node[below=\gridsize of 11] (15) {};
\node[below=\gridsize of 12] (16) {};

\node[below left=0.05cm and 0.7cm of 13] (a1) {};
\node[below left=0.05cm and 0.7cm of 9] (a2) {};
\node[left=\gridsize of a1] (a3) {};
\node[left=\gridsize of a3] (a4) {};
\node[left=\gridsize of a2] (a5) {};
\node[left=\gridsize of a4] (a6) {};
\node[above=\gridsize of a4] (a7) {};

\draw (N-1)--(1)--(N-2);
\draw (N-2)--(5)--(N-1);
\draw (a3)--(a2)--(a1)--(a3)--(a4)--(a5)--(a3);
\draw (a4)--(a6)--(a7)--(a4);
\draw (a1)--(9)--(a2);
\draw (a1)--(13)--(a2);

\draw (1)--(2)--(3)--(4);
\draw (5)--(6)--(7)--(8);
\draw (9)--(10)--(11)--(12);
\draw (13)--(14)--(15)--(16);
\draw (1)--(5)--(9)--(13);
\draw (2)--(6)--(10)--(14);
\draw (3)--(7)--(11)--(15);
\draw (4)--(8)--(12)--(16);

\end{tikzpicture}
\caption{A graph with a $2$-\wsm{} to $K_3$-free graphs (in the two shaded areas)}
\vspace{-0.5cm}
\label{fig:wsm}
\end{figure}
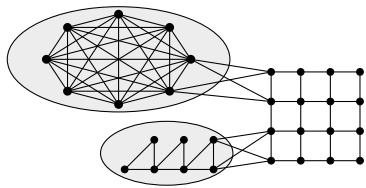

For the sake of brevity and when clear from context, we will sometimes identify $\vec{X}$ with $\bigcup_{X_i\in\vec{X}} X_i$ (for instance $G-\vec{X}$ is shorthand for $G-\bigcup_{X_i\in\vec{X}} X_i$). 
To allow a concise description of our parameters, for any hereditary graph class $\HH$ we let the \emph{well-structure number} ($\wsn^{\HH}$ in short) denote the minimum $k$ such that $G$ has a $k$-{\wsm} to $\HH$. Similarly, we let $\md^{\HH}(G)$ denote the minimum $k$ such that~$G$ has a modulator of cardinality $k$ to $\HH$. 

\sv{\begin{proposition}}
\lv{\begin{proposition}}
\label{prop:better}
Let $\HH$ be \emph{any} hereditary graph class of unbounded rank-width.
\begin{enumerate}
\item $\rw(G)\geq \wsn^\HH(G)$ for any graph $G$. Furthermore, for every $i\in \Nat$ there exists a graph $G_i$ such that $\rw(G_i)\geq \wsn^\HH(G_i)+i$, and
\item $\md^\HH(G)\geq \wsn^\HH(G)$ for any graph $G$. Furthermore, for every $i\in \Nat$ there exists a graph $G_i$ such that $\md^\HH(G_i)\geq \wsn^\HH(G_i)+i$.
\end{enumerate}
\end{proposition}

\newcommand{\pfpropbetter}[0]{
\begin{proof}
\begin{enumerate}
\item For $\rw(G)\geq \wsn^\HH(G)$ notice that for any graph $G$ of rank-width $k$, the set $\{V(G)\}$ is a $k$-{\wsm} to the empty graph. For the second claim, since $\HH$ has unbounded rank-width, for every $i\in \Nat$ it contains some graph $G_i$ such that $\rw(G_i)>i$; by definition, $\wsn^\HH(G_i)=0$.
\item For $\md^\HH(G)\geq \wsn^\HH(G)$, let $G$ be a graph containing a modulator $X=\{v_1,\dots,$ $v_k\}$ to $\HH$. It is easy to check that $\vec{X}=\{\{v_1\},\dots,\{v_k\}\}$ is a $k$-{\wsm} to $\HH$. For the second claim, let $G'\not \in \HH$ and let $k=\rw(G')$. Consider the graph $G_i$ consisting of $i+1+k$ disjoint copies of $G'$ and a vertex $q$ which is adjacent to every other vertex of $G$. Since $\HH$ is hereditary, we may assume without loss of generality that it contains the single-vertex graph. It is then easy to check that $\{V(G)\setminus \{q\}\}$ forms a $k$-{\wsm} in $G$ to $\HH$. Now consider any set $X\subseteq V(G)$ of cardinality at most $i+k$. Clearly, there must exist some copy of $G'$, say $G'_j$, such that $X\cap V(G'_j)=\emptyset$. Since $G'_j\not \in \HH$, it follows from the hereditarity of $\HH$ that $G-X\not \in \HH$ and hence $X$ cannot be a modulator to $\HH$. We conclude $\md^\HH(G_i)>i+k=i+\wsn^\HH(G_i)$.
\qed
\end{enumerate}
\end{proof}}

\lv{\pfpropbetter}

\section{Finding Well-Structured Modulators}
\label{sec:finding}

The objective of this subsection is to prove the following theorem. Interestingly, our approach only allows us to find \wsm s if the rank-width of the graph is sufficiently large. This never becomes a problem though, since on graphs with small rank-width we can always directly use rank-width as our parameter.\sv{\enlargethispage*{8mm}}

\begin{theorem}
\label{thm:main-find}
Let $\HH$ be a graph class characterized by a finite obstruction set. There exists an \emph{FPT} algorithm parameterized by $k$ which for any graph $G$ of rank-width at least $k+2$ either finds a $k$-{\wsm} to $\HH$ or correctly detects that it does not exist. 
\end{theorem}

\lsv{
We begin by stating several useful properties of splits in graphs. We remark that for most of this section we will restrict ourselves to connected graphs, and show how to deal with general graphs later on; this allows us to use the following result by Cunningham.}
{Our starting point on the path to a proof of Theorem~\ref{thm:main-find} is a theorem by Cunningham.}

\begin{theorem}[\cite{Cunningham82}]\label{thm:splitintersection}
Let $\{ A, C\}$, $\{ B, D\}$ be splits of a connected graph $G$ such that $|A\cap B|\ge 2$  and
$A \cup B\neq V(G)$. Then $\{A\cap B, C\cup D\}$ is a split of $G$.
\end{theorem}

\newcommand{\lemunion}[0]{
\begin{lemma}\label{lem:smintersection}
If $A$ and $B$ are overlapping \sm\-s of a connected graph $G=(V,E)$, then $A\cup B$ is also a \sm. Moreover, if $A\cup B\neq V$, then also $A\cap B$ is a \sm.
\end{lemma}

\begin{proof}
If $V=A\cup B$, then $A\cup B$ is clearly a \sm. 
So, assume $A\cup B\neq V$ and let $C=V\setminus A$ and $D=V\setminus B$; note that $C\cup D\neq V$ since $A,B$ are overlapping. We make the following exhaustive case distinction:
\begin{itemize}
\item if $|A\cap B|=1$ and $|C\cap D|=1$, then both $A\cap B$ and $A\cup B=V\setminus (C\cap D)$ are easily seen to be \sm s;
\item if $|A\cap B|\geq 2$ and $|C\cap D|=1$, then $A\cap B$ is a {\sm} by Theorem~\ref{thm:splitintersection} and $A\cup B$ is also a {\sm} because $C\cap D$ is a {\sm};
\item if $|A\cap B|=1$ and $|C\cap D|\geq 2$, then $A\cap B$ is a {\sm} and $A\cup B$ is also a {\sm} because $C,D$ satisfy the conditions of Theorem~\ref{thm:splitintersection} and hence $C\cap D = V\setminus (A\cup B)$ forms a \sm;
\item if $|A\cap B|\geq 2$ and $|C\cap D|\geq 2$, then $A\cap B$ is a {\sm} by Theorem~\ref{thm:splitintersection} and $A\cup B$ is also a {\sm} because $C,D$ satisfy the conditions of Theorem~\ref{thm:splitintersection}, as in the previous case.
\end{itemize}
\vspace{-0.7cm}
\qed
\end{proof}

\begin{lemma}\label{lem:diferrence}
Let $G=(V,E)$ be a connected graph and $A,B$ be overlapping \sm s. Then $A\setminus B$ is also a \sm. 
\end{lemma}

\begin{proof}
The lemma clearly holds if $|A\setminus B|\leq 1$, so we may assume that $|A\setminus B|\ge 2$.
Let $Z=V\setminus B$; since $B$ is a split module, so is $Z$. Furthermore, since $A$ and $B$ are overlapping, it holds that $B\setminus A$ is nonempty and hence $V\neq Z\cup A$. Since $Z\cap A=A\setminus B$, we have $|Z\cap A|\ge 2$ and hence we conclude that $Z\cap A=A\setminus B$ is a split module by Theorem~\ref{thm:splitintersection}.
\qed
 \end{proof}

\begin{lemma}\label{lem:unionrw}
Let $k\in \Nat$ be a constant, $G=(V,E)$ a graph, and  $A$, $B$, $C$ be pairwise disjoint \sm s such that $A\cup B\cup C=V$. Let $a$, $b$, $c$ be arbitrary vertices such that $a\in N(A)$, $b\in N(B)$, and $c\in N(C)$. 
If $\max \big(\rw(G[A\cup \{a\}]),\rw(G[B\cup \{b\}]),\rw(G[C\cup \{c\}])\big)\le k$, then $\rw(G)\le k$.
\end{lemma}

\begin{proof}
Let
  $\mathcal{T}_{A} = (T_{A}, \mu_{A})$, $\mathcal{T}_{B}
 = (T_{B},
  \mu_{B})$, and $\mathcal{T}_{C} = (T_{C}, \mu_{C})$
 be witnessing rank
  decompositions of $G[A], G[B]$, and $G[C]$,
 respectively.
 
 We construct a rank decomposition $\mathcal{T} = (T, \mu)$ of $G$ as follows.  
 
 Let $l_{a}$ be the leaf (note that $\mu_{A}$ is
 bijective) of $T_{A}$ such that $\mu_{A}(a) = l_{a}$. 
Similarly, let $l_b$ and $l_c$ be the leaves such that $\mu_B(b)=l_b$ and $\mu_C(c)=l_c$, respectively.  
 We obtain $T$ from $T_{A}$ by adding disjoint copies of
 $T_{B}$ and $T_{C}$ and then identifying $l_{a}$ with the copies of
 $l_{b}$ and $l_{b}$. Since $T_{A}, T_{B}$, and $T_{C}$ are subcubic, so
 is $T$.
 
 We define the mapping $\mu: V(G) \rightarrow \SB t
  \SM $ t
 is a leaf of $T \SE$ by
  \begin{align*}
    \mu(v) = \begin{cases}
      \mu_{a}(v) &\text{if $v \in A$,} \\
      c(\mu_{b}(v)) &\text{if $v \in B$,} \\
      c(\mu_{c}(v)) &\text{otherwise,}
    \end{cases}
  \end{align*}
  where $c$ maps internal nodes in $T_{B} \cup T_{C}$ to their copies in
  $T$. The mappings $\mu_{A}, \mu_{B}$, and $\mu_{C}$ are
  bijections and $c$ is injective, so $\mu$ is injective. By
  construction, the image of $V(G)$ under $\mu$ is the set of
  leaves of $T$, so $\mu$ is a bijection. Thus $\mathcal{T} = (T,
  \mu)$ is a rank decomposition of $G$.

We prove that the width of $\mathcal{T}$ is at most $k$. Given a rank
  decomposition $\mathcal{T}^* = (T^*, \mu^*)$ and an edge $e$ of $T^*$, the
  connected components of $T^* - e$ induce a bipartition $(X, Y)$
  of the leaves of $T^*$. We set $f: (\mathcal{T}^*, e) \mapsto
  ({\mu^*}^{-1}(X), {\mu^*}^{-1}(Y))$. Take any edge $e$ of $T$. There is a
  natural bijection $\beta$ from the edges in $T$ to the edges of $T_{A} \cup
  T_{B} \cup T_{C}$. Accordingly, we distinguish three cases for $e ' =
  \beta(e)$:
 
  \begin{enumerate}
  \item $e' \in T_{A}$. Let $(U, W) = f(\mathcal{T}_{A},
    e')$. Without loss of generality assume that $a \in
    W$.  Then by construction of $\mathcal{T}$ , we have $f(\mathcal{T},
    e) = (U, W \cup B\cup C)$. Let $u\in A$ and $v\in B\cup C$. Since $A$ is \sm~either $v\notin N(A)$ and $\mathbf{A}_G(u, v) =0$ for all $u\in A$, or $v\in N(A)$ in which case $\mathbf{A}_G(u, v) =
    \mathbf{A}_G(u, a)$ for all $u\in A$. Therefore, to obtain $\mathbf{A}_G(U, W\cup B\cup C) $ one can simply copy the column corresponding to $a$ in $\mathbf{A}_G(U, W)$ or add some empty columns. This does not increase the rank of the matrix. \label{rwmodule:case1}
  \item $e' \in T_{B}$. This case is symmetric to case~\ref{rwmodule:case1}, with
    $A$ and $B$ switching their roles and  $b$ taking the role of $a$.
  \item $e' \in T_{C}$. This case is symmetric to case~\ref{rwmodule:case1}, with
    $A$ and $C$ switching their roles and  $c$ taking the role of $a$.
  \end{enumerate}
  Since $\beta$ is bijective, this proves that the rank of any bipartite
  adjacency matrix induced by removing an edge $e \in T$ is bounded by $k$. We
  conclude that the width of $\mathcal{T}$ is at most $k$ and thus
  \mbox{$\rw(G) \leq k$}.
  \qed
\end{proof}

By repeating the proof technique of Lemma~\ref{lem:unionrw} without the set $C$, we obtain the following corollary.
\begin{corollary}
\label{cor:justtwosets}
Let $k\in \Nat$ be a constant, $G=(V,E)$ a graph, and  $A$, $B$ pairwise disjoint \sm s such that $A\cup B=V$. Let $a,b\in V$ be such that $a\in N(A)$ and $b\in N(B)$. 
If $\max \big(\rw(G[A\cup \{a\}]),\rw(G[B\cup \{b\}])\big)\le k$, then $\rw(G)\le k$.
\end{corollary}

%
\begin{lemma}\label{lem:smallrw}
 Let $k \in \Nat$ be a constant. Let $G=(V,E)$ be a connected graph and let $M_1, M_2$
  be \sm s of $G$ such that $M_1 \cup M_2 = V$ and $\max(\rw(G[M_1]),
  \rw(G[M_2])) \leq k$. Then $\rw(G) \leq k+1$.
\end{lemma}

\begin{proof}
Let $M_{22} = M_2\setminus M_1$. 
Clearly, $\{M_1, M_{22}\}$ is a split. Since rank-width is preserved by taking induced subgraphs, the graph $G[M_{22}]$ has rank-width at most $k$. Let $v_1\in N(M_{22})$ and $v_2\in N(M_1)$. It is easy to see that 
graphs  $G_{1} =
  G[M_{1} \cup \{v_{2}\}]$ and $G_{2} = G[M_{22} \cup \{v_{1}\}]$ have rank-width at most
 $k+1$. We finish the proof by applying Corollary~\ref{cor:justtwosets}, with $M_1$, $M_{22}$
 in roles of $A$, $B$ and $v_1$, $v_2$ in roles of $a$, $b$, respectively.
 \qed
\end{proof}}
\lv{\lemunion}

The following lemma in essence shows that the relation of being in a split-module of small rank-width is transitive (assuming sufficiently high rank-width). The significance of this will become clear later on.
\sv{\begin{lemma}}
\lv{\begin{lemma}}
\label{lem:union}
  Let $k \in \Nat$ be a constant. Let $G=(V,E)$ be a connected graph with rank-width at least $k+2$ and let $M_1, M_2$
  be {\sm s} of $G$ such that $M_1 \cap M_2 \neq \emptyset$ and $\max(\rw(G[M_1]),
  \rw(G[M_2])) \leq k$. Then $M_1 \cup M_2$ is a {\sm} of $G$ and $\rw(G[M_1\cup M_2]) \leq k$.
\end{lemma}

\newcommand{\pflemunion}[0]{
\begin{proof}
  If $M_1 \subseteq M_2$ or $M_2 \subseteq M_1$ the result is
  immediate, hence we may assume that they are overlapping.
  Lemma~\ref{lem:smallrw} and $\rw(G)\ge k+2$ together imply that $M_1\cup M_2 \neq V$.
Let
  $M_{11} = M_1 \setminus M_2, M_{22} = M_2 \setminus M_1$, and
 $M_{12} =
  M_1 \cap M_2$. It follows from Lemma~\ref{lem:smintersection} and Lemma~\ref{lem:diferrence} that these
  sets are {\sm s} of $G$. Let $v_{11} \in N(V\setminus M_{11}), v_{22} \in N(V\setminus M_{22})$, and
  $v_{12} \in N(V\setminus M_{12})$. We show that $\rw(G[M_1 \cup M_2]) \leq k$. By
  assumption, both
 $G[M_1]$ and $G[M_2]$ have rank-width at most $k$. Since
  rank-width
 is preserved by taking induced subgraphs, the graphs $G_{11} =
  G[M_{11} \cup \{v_{12}\}]$, $G_{12} = G[M_{12} \cup \{v_{22}\}]$,
 and
  $G_{22} = G[M_{22} \cup \{v_{12}\}]$ also have rank-width at most
 $k$. We finish the proof by applying Lemma~\ref{lem:unionrw}, with $M_{11}$, $M_{22}$, $M_{12}$
 taking the roles of $A$, $B$, and $C$ and $v_{12}$, $v_{12}$, and $v_{22}$ taking the roles of $a$, $b$, and $c$, respectively. 
 \qed
\end{proof}}

\lv{\pflemunion}

\sv{
\begin{proof}[Sketch]
The proof relies on a series of lemmas building on Theorem~\ref{thm:splitintersection}. 

  If $M_1 \subseteq M_2$ or $M_2 \subseteq M_1$ the result is
  immediate, hence we may assume that they are overlapping.
  $\rw(G)\ge k+2$ implies that $M_1\cup M_2 \neq V$. The fact that $M_1 \cup M_2$ is a {\sm} of $G$ then follows from Theorem~\ref{thm:splitintersection}. Let
  $M_{11} = M_1 \setminus M_2, M_{22} = M_2 \setminus M_1$, and
 $M_{12} =
  M_1 \cap M_2$. These sets can be shown to be {\sm s} of $G$. Let $v_{11} \in N(V\setminus M_{11}), v_{22} \in N(V\setminus M_{22})$, and
  $v_{12} \in N(V\setminus M_{12})$. We show that $\rw(G[M_1 \cup M_2]) \leq k$. By
  assumption, both
 $G[M_1]$ and $G[M_2]$ have rank-width at most $k$. Since
  rank-width
 is preserved by taking induced subgraphs, the graphs $G_{11} =
  G[M_{11} \cup \{v_{12}\}]$, $G_{12} = G[M_{12} \cup \{v_{22}\}]$,
 and
  $G_{22} = G[M_{22} \cup \{v_{12}\}]$ also have rank-width at most
 $k$. The proof can be completed by showing how the rank-decompositions of these three graphs can be combined into a rank-decomposition for $G[M_1\cup M_2]$.
\qed
\end{proof}
}
\begin{definition}
  Let $G$ be a graph and $k \in \Nat$.  We define a relation
  $\sim^G_k$ on $V(G)$ by letting $v \sim^G_k w$ if and only if there
  is a {\sm} $M$ of $G$ with $v,w \in M$ and $\rw(G[M]) \leq k$. We
  drop the superscript from $\sim^G_k$ if the graph $G$ is clear from
  context.
\end{definition}

Using Lemma~\ref{lem:union} to deal with transitivity, we prove the following.
\lv{\begin{proposition}}
\sv{\begin{proposition}}
\label{prop:eqrel}
  For every $k \in \Nat$ and graph $G=(V,E)$ with rank-width at least $k+2$, the relation $\sim_k$ is an
  equivalence relation, and each equivalence class $U$ of $\sim_k$ is a {\sm}
  of $G$ with $\rw(G[U]) \leq k$.
\end{proposition}

\newcommand{\pfpropeqrel}[0]{
\begin{proof} 
  Let $G$ be a graph and $k \in \Nat$. For every $v \in V$, the
  singleton $\{v\}$ is a {\sm} of $G$, so $\sim_k$ is reflexive. Symmetry of
  $\sim_k$ is trivial. For transitivity, let $u,v,w \in V$ be such that $u
  \sim_k v$ and $v \sim_k w$. Then there are {\sm s} $M_1, M_2$ of $G$ such
  that $u,v \in M_1$, $v, w \in M_2$, and $\rw(G[M_1]), \rw(G[M_2]) \leq k$; in particular, since $\rw(G)\geq k+2$ this implies that there exists a connected component $G'$ of $G$ containing $u,v,w$. By
  Lemma~\ref{lem:union}, $M_1 \cup M_2$ is a {\sm} of $G'$ (and hence also of $G$) such that $\rw(G[M_1 \cup
  M_2]) \leq k$. In combination with $u,w \in M_1 \cup M_2$ that implies $u
  \sim_k w$. This concludes the proof that $\sim_k$ is an equivalence
  relation.
  
  Now let $v \in V$, $G'$ be the connected component containing $v$, and let $U = [v]_{\sim_k}$. For each $u \in U$ there is
  a {\sm} $W_u$ of $G'$ (and of $G$) with $u,v \in W_u$ and $\rw(G[W_u]) \leq k$. By
  Lemma~\ref{lem:union}, $W = \bigcup_{u \in U} W_u$ is a {\sm} of $G'$ (and hence also of $G$) and $\rw(G[W]) \leq k$. Clearly, $[v]_{\sim_k} \subseteq W$. On the other hand,
  $u \in W$ implies $v \sim_k u$ by definition of $\sim_k$, so $W \subseteq
  [v]_{\sim_k}$. That is, $W = [v]_{\sim_k}$.
\end{proof}}
\lv{\pfpropeqrel}

\begin{corollary}
\label{cor:equiv}
Any graph $G$ of rank-width at least $k+2$ has its vertex set uniquely partitioned by the equivalence classes of $\sim_k$ into inclusion-maximal \sm s of rank-width at most $k$.
\end{corollary}

\lv{Next, we state a simple but useful observation.}

\newcommand{\obsdisconequiv}[0]{
\begin{observation}
\label{obs:disconequiv}
Let $k\in \Nat$, $G$ be a disconnected graph with rank-width at least $k+2$, and $\CC(G)$ be the set of connected components of $G$. Then $\sim^G_k=\bigcup_{G' \in \CC(G)}\sim^{G'}_k$.
\end{observation}
}
\lv{\obsdisconequiv}

Now that we know $\sim_k$ is an equivalence, we show how to compute it in FPT time.

\lv{\begin{proposition}}
\sv{\begin{proposition}}
\label{prop:equdecision}
  Let $k \in \Nat$ be a constant. Given an $n$-vertex graph $G$ of rank-width at least $k+2$ and two vertices $v, w$, we can decide whether $v \sim_k w$ in time $\bigoh(n^3)$.
\end{proposition}

\newcommand{\pfpropequdecision}[0]{
\begin{proof}
From Observation~\ref{obs:disconequiv} it follows that if the proposition holds for connected graphs, then it holds for disconnected graphs as well; hence we may assume that $G$ is connected. By Theorem~\ref{thm:computingGLT} we can compute the unique split-tree $ST(G)=(T,\FF)$ in 
$O(m+n)\alpha(m+n)$ time. Due to Theorem~\ref{thm:splitinGLT}, every split in $G$ is the bipartition of leaves of $T$ induced either by removing an internal tree-edge of $T$ or an edge created by a node-split of a degenerate vertex of $T$. 

Vertices of $G$ are leaves of $T$ and we can find a path $P$ between $v$ and $w$ in $T$ in time linear in size of $T$. There are at most linearly many vertices on the path and we can split every degenerate vertex on $P$ in a way that every degenerate vertex on a new path $P'$ between $u$ and $v$ will have $3$ vertices. Denote the new tree by $T'$.

Now every edge between $P'$ and $T'\setminus P'$ corresponds to a minimal \sm~containing $v$ and $w$. Conversely, as a consequence of Theorem~\ref{thm:splitinGLT} every minimal {\sm} containing $v$ and $w$ is induced by removing an edge between $P'$ and $T'\setminus P'$, and let $M_{vw}$ be the set containing all of these at most $|T|$ minimal split modules. Hence, $v \sim_k w$ if and only if there is a {\sm} $X$ in $M_{vw}$ such that $\rw(G[X])\leq k$.
By Theorem~\ref{thm:rankdecomp} we can decide, for each such $X$, whether $\rw(G[X])\le k$ in time $f(k)\cdot n^3$, where $f$ is some computable function.
\qed
\end{proof}}

\lv{\pfpropequdecision}

\sv{\begin{proof}[Sketch]
The definition of \sm s allows us to consider each connected component of a graph separately. We then compute the so-called \emph{split-tree}~\cite{Cunningham82,GioanPaul07,GioanPaul12,GioanPaulTedderCorneil14} of $G$ and use it to list all minimal \sm s containing $v$ and $w$. Finally, we check whether any of these \sm s has rank-width at most $k$ by using Theorem~\ref{thm:rankdecomp}.
\qed
\end{proof}
}

\lv{In the rest of this section we show how to find a $k$-\wsm{} to any graph class $\HH$ characterized by a finite obstruction set $\FF$. We first present the algorithm and then show its running time and correctness.}
\sv{We are now ready to present an algorithm for finding a $k$-\wsm{} to any graph class $\HH$ characterized by a finite obstruction set $\FF$.}

\medskip
\RestyleAlgo{boxruled}
\begin{algorithm}[H]
\LinesNumbered

\SetKwInOut{Input}{Input}\SetKwInOut{Output}{Output}
\SetKwInOut{Parameter}{Parameter}

 \Input{$k\in \Nat_0$, $n$-vertex graph $G$, equivalence $\sim$ over a superset of $V(G)$}
 \Output{A $k$-cardinality set $\vec{X}$ of subsets of $V(G)$, or \emph{False}}


\BlankLine
  \uIf{$G$ does not contain any $D\in \FF$ as an induced subgraph}{\KwRet{$\emptyset$}}
  \Else{$D':=$ an induced subgraph of $G$ isomorphic to an arbitrary $D\in \FF$\;}
 \lIf{$k=0$}{\KwRet{False}}

\ForEach{$[a]_{\sim}$ of $G$ which intersects with $V(D')$}{
	$\vec{X}=\text{FindWSM}_\FF(k-1, G-[a]_{\sim}, \sim)$\;\label{algstep:recursion}
	\If{$\vec{X}\neq$ \emph{False}}{
	\KwRet{$\vec{X}\cup \{[a]_{\sim}\}$}	
	}
}

\KwRet{False}
 \caption{FindWSM$_\FF$} \label{alg:findingmodulator}
\end{algorithm}
 \medskip
 
We will use $\sim_k$ as the input for \emph{FindWSM}$_\FF$, however considering general equivalences as inputs is useful for proving correctness.
\lv{Recall that the equivalence $\sim_k$ (or, more precisely, the set of its equivalence classes) can be computed in time $n^2\cdot f(k)\cdot n^3$ for some function $f$ thanks to Proposition~\ref{prop:equdecision}, and this only needs to be done once before starting the algorithm.} \lv{The following two lemmas show that Algorithm~\ref{alg:findingmodulator} is correct and runs in FPT time.}

\sv{
\begin{lemma}
\label{lem:algo}
There exists a constant $c$ such that \emph{FindWSM}$_\FF$ runs in time $c^k\cdot n^{\bigoh(1)}$. Furthermore, if $G$ is a graph of rank-width at least $k+2$ and $\sim_k$ is the equivalence computed by Proposition~\ref{prop:equdecision}, then \emph{FindWSM}$_\FF(k,G,\sim_k)$ outputs a $k$-wsm{} to $\HH$ or correctly detects that no such $k$-wsm exists in $G$.
\end{lemma}
}

\newcommand{\algo}[0]{
\begin{lemma}
\label{lem:algoruntime}
There exists a constant $c$ such that \emph{FindWSM}$_\FF$ runs in time $c^k\cdot n^{\bigoh(1)}$.
\end{lemma}
\begin{proof}
The time required to perform the steps on rows {$\mathbf{2}$-$\mathbf{6}$} is $n^{O(1)}$ since $\FF$ is finite. For the same reason, it holds that $|V(D')|$ and hence also the number of times the procedure on rows {$\mathbf{8}$-$\mathbf{13}$} is called are bounded by a constant, say $c$ (to be precise, $c$ is bounded by the order of the largest graph in $\FF$). 

For the rest of the proof, we proceed by induction on $k$.
First, if $k=0$, then the algorithm is polynomial by the above. So assume that $k\ge 1$ and the algorithm for $k-1$ runs in time at most $c^{k-1}\cdot n^{O(1)}$. Then the algorithm for $k$ will run in polynomial time up to rows {$\mathbf{8}-\mathbf{13}$}, where it will make at most $c$ calls to the algorithm for $k-1$, which implies that the running time for $k$ is bounded by $c^{k}\cdot n^{O(1)}$. \qed
\end{proof}

\begin{lemma}
\label{lem:algocorrect}
Let $k\geq 0$, $G=(V,E)$ be a graph and $\sim$ an equivalence over a superset of $V$. Then
\emph{FindWSM}$_\FF(k,G,\sim)$ outputs a set $\vec{X}$ of at most $k$ equivalence classes of $\sim$ such that $G-\vec{X}$ is $\FF$-free. 
\end{lemma}

\begin{proof}
If $G$ does not contain any $D$ as an induced subgraph, then we correctly return the empty set. So, assume there exists an induced subgraph $D'$ of $G$ isomorphic to $D$.  
We prove the lemma by induction on $k$. 

Clearly, if $k=0$ but there exists some obstruction, then the algorithm outputs \emph{False} and this is correct; if $k=0$ and no obstruction exists, then the algorithm correctly outputs $\emptyset$. 
Let $k\ge 1$ and assume that the algorithm is correct for $k-1$. If $G$ does not contain any such $\vec{X}$, then for any equivalence class $[a]_\sim$, FindWSM$_\FF(k-1, G-[a]_{\sim}, \sim)$ will correctly output \emph{False}. 

On the other hand, assume $G$ does contain some $\vec{X}$ with the desired properties. In particular, this implies that $\vec{X}$ must intersect $V(D')$. Let $X_i$ be an arbitrary equivalence class of $\vec{X}$ which intersects $V(D')$. Then $\vec{X}'\setminus \{X_i\}$ is a set of at most $k-1$ equivalence classes of $\sim$ in $G-X_i$, and hence FindWSM$_\FF(k-1, G-X'_i, \sim)$ will output some solution $\vec{X}''$ for $G-X'_i$ by our inductive assumption. Since any obstruction in $G$ intersecting $X'_i$ is removed by $X'_i$ and $G-X'_i$ is made $\FF$-free by $\vec{X}''$, we observe that $\vec{X}''\cup X'_i$ intersects every obstruction in $G$ and hence the proof is complete.
\qed
\end{proof}

From Lemma~\ref{lem:algocorrect} and Corollary~\ref{cor:equiv} we obtain the following.
 
\begin{corollary}
\label{cor:algocorrect}
Let $k\in \Nat$, $G$ be a graph of rank-width at least $k+2$ and $\sim_k$ be the equivalence computed by Proposition~\ref{prop:equdecision}. Then
\emph{FindWSM}$_\FF(k,G,\sim_k)$ outputs a $k$-wsm{} to $\HH$ or correctly detects that no such $k$-wsm exists in $G$.
\end{corollary}}

\lv{\algo}

\begin{proof}[of Theorem~\ref{thm:main-find}]
\lv{The theorem follows by using Proposition~\ref{prop:equdecision} and then Algorithm~\ref{alg:findingmodulator} in conjunction with Lemma~\ref{lem:algoruntime} and~\ref{lem:algocorrect}.}
\sv{The theorem follows by using Proposition~\ref{prop:equdecision} and then Algorithm~\ref{alg:findingmodulator} in conjunction with Lemma~\ref{lem:algo}.}
\qed
\end{proof}

\section{Examples of Algorithmic Applications}
\label{sec:examples}

In this section, we show how to use the notion of $k$-\wsm s to design efficient parameterized algorithms for two classical NP-hard graph problems, specifically \textsc{Minimum Vertex Cover (MinVC)} and \textsc{Maximum Clique (MaxClq)}. 
Given a graph $G$, we call a set $X\subseteq V(G)$ a \emph{vertex cover} if every edge is incident to at least one $v\in X$ and a \emph{clique} if $G[X]$ is a complete graph.

\vspace{-.5cm}
\begin{center}
  \begin{boxedminipage}[t]{0.99\textwidth}
\begin{quote}
\smallskip
  \textsc{MinVC, MaxClq}\\ \nopagebreak
  \emph{Instance}: A graph $G$ and an integer $m$.\\ \nopagebreak
  \emph{Task} (\textsc{MinVC}): Find a vertex cover in $G$ of cardinality at most $m$, or determine that it does not exist. \\
  \emph{Task} (\textsc{MaxClq}): Find a clique in $G$ of cardinality at least $m$, or determine that it does not exist. \smallskip
\end{quote}
\end{boxedminipage}
\end{center}

Establishing the following theorem is the main objective of this section.

\begin{theorem}
\label{thm:problems}
Let $\PP\in\{\textsc{MinVC},\textsc{MaxClq}\}$ and $\HH$ be a graph class characterized by a finite obstruction set. Then $\PP$ is \emph{FPT} parameterized by $\wsn^\HH$ if and only if $\PP$ is polynomial-time tractable on $\HH$.
\end{theorem}

Since $\wsn^\HH(G)=0$ for any $\FF$-free graph $G$, the ``only if'' direction is immediate; in other words, being polynomial-time tractable on $\HH$ is clearly a necessary condition for being fixed parameter tractable when parameterized by $\wsn^\HH(G)$. Below we prove that for the selected problems this condition is also sufficient.

\lv{\begin{lemma}}
\sv{\begin{lemma}}
\label{lem:vc}
If \textsc{MinVC} is polynomial-time tractable on a graph class $\HH$ characterized by a finite obstruction set, then $\textsc{MinVC}[\wsn^\HH]$ is \emph{FPT}.
\end{lemma}

\newcommand{\pflemvc}[0]{
\begin{proof}
Let $G=(V,E)$ be a graph and let $k=\wsn^\HH(G)$. If $\rw(G)\leq k+2$, then we simply use known algorithms to solve the problem in FPT time~\cite{GanianHlineny10}. Otherwise, we proceed by using Theorem~\ref{thm:main-find} to compute a $k$-{\wsm} $\vec{X}=\{X_1,\dots,X_k\}$ in FPT time. For each $i\in [k]$, we let $A_i$ be the frontier of $X_i$ and we let $B_i=N(A_i)$.

Since for each $i\in [k]$ the graph $G[A_i\cup B_i]$ contains a complete bipartite graph, any vertex cover of $G$ must be a superset of either $A_i$ or $B_i$. We can branch over these options for each $i$ in $2^k$ time; formally, we branch over all of the at most $2^k$ functions $f:[i]\rightarrow \{A,B\}$, and refer to these as \emph{signatures}. Each vertex cover $Y$ of $G$ can be associated with at least one signature $f$, constructed in the following way: for each $i\in [k]$ such that $A_i\subseteq Y$, we set $f(i)=A$, and otherwise we set $f(i)=B$.

Our algorithm then proceeds as follows. For a graph $G$ and a signature $f$, we construct a partial vertex cover $Z=\bigcup_{i\in[k]} f(i)$. We let $G'=G-Z$. Consider any connected component $C$ of $G'$. If $C$ intersects some $X_i$, then by the construction of $Z$ it must hold that $C\subseteq X_i$. Hence it follows that $C$ either has rank-width at most $k$ (in the case $C\subseteq X_i$ for some $i$), or $C$ is in $\HH$ (if $C$ does not intersect $\vec{X}$), or both. Then we find a minimum vertex cover for each connected component of $G'$ independently, by either calling the known FPT algorithm (if $C$ has bounded rank-width) or the polynomial algorithm (if $C$ is in $\HH$) at most $|C|$ times. Let $Z'$ be the union of the obtained minimum vertex covers over all the components of $G'$, and let $Y_f=Z\cup Z'$. After branching over all possible functions $f$, we compare the obtained cardinalities of $Y_f$ and choose any $Y_f$ of minimum cardinality. Finally, we compare $|Y_f|$ and the value of $m$ provided in the input.

We argue correctness in two steps. First, assume for a contradiction that $G$ contains an edge $e$ which is not covered by $Y_f$ for some $f$. Then $e$ cannot have both endpoints in $G'$, since $Y_f$ contains a (minimum) vertex cover for each connected component of $G'$, but $e$ cannot have an endpoint outside of $G'$, since $Z\subseteq Y_f$. Hence each $Y_f$ is a vertex cover of $G$.

Second, assume for a contradiction that there exists a vertex cover $Y'$ of $G$ which has a lower cardinality than the vertex cover found by the algorithm described above. Let $f$ be the signature of $Y'$. Then it follows that $Z\subseteq Y'$, and since $Z\subseteq Y_f$, there would exist a component $C$ of $G\setminus Z$ such that $|Y'\cap C|\leq |Y_f\cap C|$. However, this would contradict the minimality of $Z'\cap C=Y_f\cap C$. Hence we conclude that no such $Y'$ can exist, and the algorithm is correct.
\qed
\end{proof}}
\lv{\pflemvc}

\sv{\begin{proof}[Sketch]
We compute a $k$-\wsm{} $\vec{X}$ to $\HH$ in $G$ by Theorem~\ref{thm:main-find}. For each element $X_i\in \vec{X}$, it holds that either the frontier of $X_i$ or its neighborhood in $G-X_i$ must be in any vertex cover of $G$. Branching on these at most $2^k$ options allows us to reduce the instance to at most $2^k$ disconnected instances such that each connected component has either rank-width bounded by $k$ or is in $\HH$; these connected components can then be solved independently.
\qed
\end{proof}
}



\lv{We deal with the second problem below.}

\sv{\begin{lemma}}
\lv{\begin{lemma}}
\label{lem:clq}
If \textsc{MaxClq} is polynomial-time tractable on a graph class $\HH$ characterized by a finite obstruction set, then $\textsc{MaxClq}[\wsn^\HH]$ is \emph{FPT}.
\end{lemma}

\newcommand{\pflemclq}[0]{
\begin{proof}
We begin in the same way as for \textsc{MinVC}: let $G=(V,E)$ be a graph and let $k=\wsn^\HH(G)$. If $\rw(G)\leq k+2$, then we simply use known algorithms to solve the problem in FPT time~\cite{GanianHlineny10}. Otherwise, we proceed by using Theorem~\ref{thm:main-find} to compute a $k$-{\wsm} $\vec{X}=\{X_1,\dots,X_k\}$ in FPT time. For each $i\in [k]$, we let $A_i$ be the frontier of $X_i$ and we let $B_i=N(A_i)$.

Let $X_0=G-\vec{X}$ and let $s\subseteq \{0\}\cup [k]$. Then any clique $C$ in $G$ can be uniquely associated with a \emph{signature} $s$ by letting $i\in s$ if and only if $X_i\cap C\neq \emptyset$. The algorithm proceeds by branching over all of the at most $2^{k+1}$ possible non-empty signatures $s$. If $|s|=1$, then the algorithm simply computes a maximum-cardinality clique in $X_{s}$ (by calling the respective FPT or polynomial algorithm at most a linear number of times) and stores it as $Y_{s}$. 

If $|s|\geq 2$, then the algorithm makes two checks before proceeding. First, if $0\in s$ then it constructs the set $X'_0$ of all vertices $x\in X_0$ such that $x$ is adjacent to every $A_i$ for $i\in s\setminus \{0\}$. If $X'_0=\emptyset$ then the current choice of $s$ is discarded and the algorithm proceeds to the next choice of $s$. Second, for every $a\neq b$ such that $a,b\in s\setminus \{0\}$ it checks that $X'_a=A_a$ and $X'_b=A_b$ are adjacent; again, if this is not the case, then we discard this choice of $s$ and proceed to the next choice of $s$. Finally, if the current choice of $s$ passed both tests then for each $i\in s$ we compute a maximum clique in each $G[X'_i]$ and save their union as $Y_s$. In the end, we choose a maximum-cardinality set $Y_s$ and compare its cardinality to the value of $m$ provided in the input.

We again argue correctness in two steps. First, assume for a contradiction that $Y_s$ is not a clique, i.e., there exist distinct non-adjacent $a,b\in Y_s$. Since $Y_s$ consists of a union of cliques within subsets of $X'_{i\in s}$, it follows that there would have to exist distinct $c,d\in s$ such that $a\in X'_c$ and $b\in X'_d$. This can however be ruled out for $c$ or $d$ equal to $0$ by the construction of $X'_0$. Similarly, if $c$ and $d$ are both non-zero, then this is impossible by the second check which tests adjacency of every pair of $X'_c$ and $X'_d$ for every $c,d\in s$.

Second, assume for a contradiction that there exists a clique $Y'$ in $G$ which has a higher cardinality than the largest clique obtained by the above algorithm. Let $s$ be the signature of $Y'$. If $|s|=1$ then $|Y_s|\geq |Y'|$ by the correctness of the respective FPT or polynomial algorithm used for each $X_s$. If $|s|\geq 2$ then $Y'$ may only intersect the sets $X'$ constructed above for $s$. Moreover, if there exists $i\in [k]\cup \{0\}$ such that $|Y'\cap X'_i|> |Y_s\cap X'_i|$ then we again arrive at a contradiction with the correctness of the respective FPT or polynomial algorithms used for $X'_i$. Hence we conclude that no such $Y'$ can exist, and the algorithm is correct.
\qed
\end{proof}}
\lv{\pflemclq}

Finally, let us review some concrete graph classes for use in Theorem~\ref{thm:problems}.
\lv{We use $K_i$, $C_i$ and $P_i$ to denote the $i$-vertex complete graph, cycle, and path, respectively. $2K_2$ denotes the disjoint union of two $K_2$ graphs, and the \emph{fork} graph is depicted for instance in~\cite{Alekseev04}. The $K_{3,3}\text{-}e$, \emph{banner}, \emph{twin-house} and $T_{2,2,2}$ graphs are defined in~\cite{BrandstadtLozin01,GerberLozin03}.}

\sv{\begin{ourfact}}
\lv{\begin{ourfact}}
  \label{fact:vcistract}
\textsc{MinVC} is polynomial-time tractable on the following graph classes:
\begin{enumerate}
\item $(2K_2,C_4,C_5)$-free graphs (split graphs);
\item $P_5$-free graphs\sv{~{\normalfont\cite{LokshtanovVatshelleVillanger14}}};
\item fork-free graphs\sv{~{\normalfont\cite{Alekseev04}}};
\item $(\text{banner}, T_{2,2,2})$-free graphs and $(\text{banner}, K_{3,3}\text{-}e, \text{twin-house})$-free graphs\sv{~{\normalfont\cite{GerberLozin03,BrandstadtLozin01}}}.
\end{enumerate}
\end{ourfact}

\newcommand{\pffactvcistract}[0]{
\begin{proof}
\begin{enumerate}
\item Split graphs are graphs whose vertex set can be partitioned into one clique and one independent set, and this partitioning can be found in linear time. If each vertex in the clique is adjacent to at least one independent vertex, then the clique is a minimum vertex cover, otherwise the clique without a pendant-free vertex is a minimum vertex cover.
\item See \cite{LokshtanovVatshelleVillanger14}.
\item See \cite{Alekseev04}.
\item See \cite{GerberLozin03} and \cite{BrandstadtLozin01}. \qed
\end{enumerate}
\end{proof}}
\lv{\pffactvcistract}

\sv{\begin{ourfact}}
\lv{\begin{ourfact}}
\label{fact:clqistract}
\textsc{MaxClq} is polynomial-time tractable on the following graph classes:
\begin{enumerate}
\item Any complementary graph class to the classes listed in Fact \ref{fact:vcistract} (such as cofork-free graphs and split graphs);
\item Graphs of bounded degree.
\end{enumerate}
\end{ourfact}

\newcommand{\pffactclqistract}[0]{
\begin{proof}
\begin{enumerate}
\item It is well-known that each maximum clique corresponds to a maximum independent set (and vice-versa) in the complement graph.
\item The degree bounds the size of a maximum clique, again resulting in a simple folklore branching algorithm. The class of graphs of degree at most $d$ is exactly the class of $\FF$-free graphs for $\FF$ containing all $(d+1)$-vertex supergraphs of the star with $d$ leaves. \qed
\end{enumerate}
\end{proof}}
\lv{\pffactclqistract}

\section{$\MSO$ Model Checking with Well-Structured Modulators}
\label{sec:mso}

Here we show how well-structured modulators can be used to solve the MSO Model Checking problem, as formalized in Theorem~\ref{thm:main-use} below.
Note that our meta-theorem captures not only the generality of MSO model checking problems, but also applies to a potentially unbounded number of choices of the graph class $\HH$. Thus, the meta-theorem supports two dimensions of generality.

\begin{theorem}\label{thm:main-use}
For every $\MSO$ sentence $\phi$ and every graph class $\HH$ characterized by a finite obstruction set such that $\MSOMC{\phi}$ is \emph{FPT} parameterized by $\md^\HH(G)$, the problem $\MSOMC{\phi}$ is \emph{FPT} parameterized by $\wsn^\HH(G)$.
\end{theorem}
The condition that $\MSOMC{\phi}$ is FPT parameterized by $\md^\HH(G)$ is a necessary condition for the theorem to hold by Proposition~\ref{prop:better}. However, it is natural to ask whether it is possible to use a weaker necessary condition instead, specifically that $\MSOMC{\phi}$ is polynomial-time tractable in the class of $\FF$-free graphs (as was done for specific problems in Section~\ref{sec:examples}). Before proceeding towards a proof of Theorem~\ref{thm:main-use}, we make a digression and show that the weaker condition used in Theorem~\ref{thm:problems} is in fact not sufficient for the general case of MSO model checking. 

\lv{\begin{lemma}}
\sv{\begin{lemma}}
\label{lem:MSOhard}
There exists an MSO sentence $\phi$ and a graph class $\HH$ characterized by a finite obstruction set such that $\MSOMC{\phi}$ is polynomial-time tractable on $\HH$ but NP-hard on the class of graphs with $\wsn^\HH(G)\leq 2$ or even $\md^\HH(G)\leq 2$.
\end{lemma}

\newcommand{\pflemMSOhard}[0]{
\begin{proof}
Consider the sentence $\phi$ which describes the existence of a proper $5$-coloring of the vertices of $G$, and let $\HH$ be the class of graphs of degree at most $4$ (in other words, let $\FF$ contain all $6$-vertex supergraphs of the star with $5$ leaves). There exists a trivial greedy algorithm to obtain a proper $5$-coloring of any graph of degree at most $4$, hence $\MSOMC{\phi}$ is polynomial-time tractable on $\HH$. Now consider the class of graphs obtained from $\HH$ by adding, to any graph in $\HH$, two adjacent vertices $y,z$ which are both adjacent to every other vertex in the graph. By construction, any graph $G'$ from this new class satisfies $\md^\HH(G')\leq 2$ and hence also $\wsn^\HH(G')\leq 2$. However, $G'$ admits a proper $5$-coloring if and only if $G'-\{y,z\}$ admits a proper $3$-coloring. Testing $3$-colorability on graphs of degree at most $4$ is known to be NP-hard~\cite{KocholLozinRanderath03}, and hence the proof is complete.
\qed
\end{proof}}
\lv{\pflemMSOhard}

\sv{\begin{proof}[Sketch]
Let $\phi$ describe vertex $5$-colorability and let $\HH$ be the class of graphs of degree at most $4$. Now consider the class of graphs obtained from $\HH$ by adding two adjacent vertices $y,z$ which are adjacent to every other vertex. Hardness follows from hardness of $3$-colorability on graphs of degree at most $4$~\cite{KocholLozinRanderath03}.
\qed
\end{proof}
}

Our strategy for proving Theorem~\ref{thm:main-use} relies on a replacement technique, where each split-module in the \wsm{} is replaced by a small representative. We use the notion of \emph{similarity} defined below to prove that this procedure does not change the outcome of \MSOMC{\varphi}.

\begin{definition}[Similarity]\label{def:similarity}
  Let $q$ and $k$ be non\hy negative integers, $\HH$ be a graph class, and let $G$ and $G'$ be graphs with $k$-\wsm s $\vec{X}=\{X_1,\dots,X_k\}$ and $\vec{X'}=\{X_1',\dots,X_k'\}$ to $\HH$,
  respectively. For $1\leq i\leq k$, let $S_i$ contain the frontier of split module $X_i$ and similarly let $S'_i$ contain the frontier of split module $X'_i$.
We say that $(G, \vec{X})$ and $(G', \vec{X}')$ are \emph{$q$-similar}
  if all of the following conditions are met:
  \begin{enumerate}
  \item There exists an isomorphism $\tau$ between $G-\vec{X}$ and $G'-\vec{X}'$. \label{cond:sameH}
  \item For every $v\in V(G)\setminus \vec{X}$ and $i\in [k]$, it holds that $v$ is adjacent to $S_i$ if and only if $\tau(v)$ is adjacent to $S'_i$.\label{cond:Hcon}
  \item if $k\geq 2$, then for every $1\leq i<j\leq k$ it holds that $S_i$ and $S_j$ are adjacent if and only if $S'_i$ and $S'_j$ are adjacent. \label{cond:splitcon}
  \item For each $i\in [k]$, it holds that $\mathit{type}_q(G[X_i],S_i)=\mathit{type}_q(G'[X'_i],S'_i)$. \label{cond:sametypes}
\end{enumerate}
\end{definition}

\lv{\begin{lemma}}
\sv{\begin{lemma}}
\label{lem:partitiongame}
  Let $q$ and $k$ be non\hy negative integers, $\HH$ be a graph class, and let $G$ and $G'$ be graphs with $k$-\wsm s 
  $\vec{X}=\{X_1,\dots,X_k\}$ and $\vec{X'}=\{X_1',\dots,X_k'\}$ to $\HH$,
  respectively. If $(G, \vec{X})$ and $(G', \vec{X}')$ are \emph{$q$-similar}, then 
  $\mathit{type}_q(G, \emptyset) = \mathit{type}_q(G',
  \emptyset)$.
  \end{lemma}
  
\newcommand{\pflempartitiongame}[0]{
\begin{proof}
  For $i \in [k]$, we write $G_i = G[X_i]$ and $G_i' = G'[X_i']$. Let $X_0=V(G)\setminus \vec{X}$ and $X'_0=V(G')\setminus \vec{X'}$. By Theorem~\ref{thm:msogames}, Condition~\ref{cond:sametypes} of
  Definition~\ref{def:similarity} is equivalent to $(G_i, S_i)
  \equiv^{\MSO}_q (G'_i, S'_i)$. That is, for each $i \in [k]$,
  Duplicator has a winning strategy $\pi_i$ in the $q$\hy round \MSO game
  played on $G_i$ and $G_i'$ starting from $(S_i, S'_i)$. We
  construct a strategy witnessing $(G, \emptyset) \equiv^{\MSO}_q (G',
  \emptyset)$ in the following way:
    \begin{enumerate}
    \item Suppose Spoiler makes a set move $W$ and assume without loss of
      generality that $W \subseteq V(G)$. For $i \in [k]$, let $W_i = X_i \cap
      W$, and let $W_i'$ be Duplicator's response to $W_i$ according to
      $\pi_i$. Furthermore, let $W'_0=\SB \tau(v)\SM v\in W\cap X_0 \SE$. Then Duplicator responds with $W' = W'_0 \cup \bigcup_{i=1}^k W_i'$.
    \item Suppose Spoiler makes a point move $s$ and again assume without loss
      of generality that $s \in V(G)$. If $s \in X_i$ for some $i \in
      [k]$, then Duplicator responds with $s' \in X_i'$ according to $\pi_i$; otherwise, Duplicator responds with $\tau(s)$ as per Definition~\ref{def:similarity} point~\ref{cond:sameH}.
    \end{enumerate}
    Assume Duplicator plays according to this strategy and consider a play of
    the $q$\hy round \MSO game on $G$ and $G'$ starting from $(\emptyset, \emptyset)$. Let $\vec{v}=(v_1,\dots,v_m)$ and
    $\vec{u}=(u_1,\dots,u_m) $ be the point moves in $V(G)$ and $V(G')$ respectively, and let $\vec{V}=(V_1, \dots, V_l)$ and $\vec{U}=(U_1,\dots, U_l)$ be the set
    moves in $V(G)$ and $V(G')$ respectively, so that $l + m = q$ and the moves made in the same round have the same index. We claim that $(\vec{v},\vec{u})$
    defines a partial isomorphism between $(G,\vec{V})$ and $(G',\vec{U})$.
    \begin{itemize}
    \item Let $j_1,j_2 \in [m]$ and let $v_{j_1}, v_{j_2}\in X_0$. Since $\tau$ is an isomorphism as per Definition~\ref{def:similarity} point~\ref{cond:sameH}, it follows that $v_{j_1}=v_{j_2}$ if and only if $u_{j_1} = u_{j_2}$ and $v_{j_1}v_{j_2}\in E(G)$ if and only if $u_{j_1}u_{j_2}\in E(G')$.
    
    \item Let $j_1,j_2 \in [m]$ and let $i\in [k]$ be such that $v_{j_1}\in X_0$ and $v_{j_2}\in X_i$. Then clearly $v_{j_1}\neq v_{j_2}$ and $u_{j_1}\neq u_{j_2}$. Consider the case $v_{j_1}v_{j_2}\in E(G)$. Then $v_{j_2}$ must lie in the frontier of $X_i$, and hence $v_{j_2}\in S_i$. Since Duplicator's strategy $\pi_i$ is winning for $(G_i,S_i)$ and $(G'_i,S'_i)$, it must hold that $u_{j_2}\in S'_i$. By Definition~\ref{def:similarity} point~\ref{cond:Hcon}, it then follows that $\tau(v_{j_1})u_{j_2}\in E(G')$. So, consider the case $v_{j_1}v_{j_2}\not \in E(G)$. Then either $v_{j_2}\not \in S_i$, in which case it holds that $u_{j_2}\not \in S'_i$ because of the choice of $\pi_i$ and hence there cannot be an edge $u_{j_2}u_{j_1}$ in $G'$, or $v_{j_2}\in S_i$, in which case it holds once again that $u_{j_2}u_{j_1}\not \in E(G')$ by Definition~\ref{def:similarity} point~\ref{cond:Hcon}.
    
    \item Let $j_1,j_2 \in [m]$ and let $i \in [k]$ be such that $v_{j_1},v_{j_2}
      \in X_{i}$. Since
      Duplicator plays according to a winning strategy $\pi_i$ in the game on $G_i$
      and $G_i'$, the restriction $(\vec{v}|_i, \vec{u}|_i)$ defines a partial
      isomorphism between $(G_i, (\vec{V})|_i)$ and $(G_i',
      (\vec{U})|_i)$. It follows that $(v_{j_1},v_{j_2}) \in E(G)$ if
      and only if $(u_{j_1},u_{j_2}) \in E(G')$ and $v_{j_1} = v_{j_2}$ if and
      only if $u_{j_1} = u_{j_2}$.  
    
      \item Let $j_1,j_2 \in [m]$ and let $i_1,i_2 \in [k]$ be pairwise distinct numbers such that $v_{j_1}\in X_{i_1}$ and $v_{j_2}\in X_{i_2}$. Then $v_{j_1}
      \neq v_{j_2}$ and also $u_{j_1} \neq u_{j_2}$ since $u_{j_1} \in 
      X_{i_1}'$ and $u_{j_2} \in X_{i_2}'$ by the Duplicator's strategy. Suppose $v_{j_1}v_{j_2}\in E(G)$. Then $v_{j_1}\in S_{i_1}$, and $v_{j_2}\in S_{i_2}$, and $S_{i_1}$ and $S_{i_2}$ are adjacent in $G$. From the correctness of $\pi_{i_1}$ and $\pi_{i_2}$ it follows that $u_{j_1}\in S'_{i_1}$ and $u_{j_2}\in S'_{i_2}$, and from Definition~\ref{def:similarity} point~\ref{cond:splitcon} it follows that $S'_{i_1}$ and $S'_{i_2}$ are adjacent in $G'$, which together implies  $u_{j_1}u_{j_2}\in E(G')$. On the other hand, suppose $v_{j_1}v_{j_2}\not \in E(G)$. Then either $v_{j_1}\not \in S_{i_1}$, or $v_{j_2}\not \in S_{i_2}$, or $S_{i_1}$ and $S_{i_2}$ are not adjacent in $G$. In the first case we have $u_{j_1}\not \in S'_{i_1}$, in the second case we have $u_{j_2}\not \in S'_{i_2}$, and in the third case it holds that $S'_1$ and $S'_2$ are not adjacent in $G'$; any of these three cases imply $u_{j_1}u_{j_2}\not \in E(G')$.
    
    \item Let $j \in [m]$ such that $v_j\in X_0$. Then by the Duplicator's strategy on $X_0$ it follows that for any $V_q$ such that $v_j\in V_q$ it holds that $u_j\in U_q$ and for any $V_q$ such that $v_j\not \in V_q$ it holds that $u_j\not \in U_q$.
    
    \item Let $j\in [m]$ and $i\in [k]$ such that $v_j\in X_k$. Let $V_q$ be such that $v_j\in V_q$. Since $\pi_i$ is a winning strategy for Duplicator, it must be the case that $u_j\in U_q$. Similarly, if $v_j\not\in V_q$ then the correctness of $\pi_i$ guarantetes that $u_j\not\in U_q$.\qed
  \end{itemize}
\end{proof}}

\lv{\pflempartitiongame}

\sv{
\begin{proof}[Sketch]
The proof argument uses the $q$-round MSO game defined, e.g., in~\cite{Libkin04}. The notion of $q$-similarity ensures that the Duplicator has a winning strategy on $G'$, which translates to $G$ and $G'$ having the same $type_q$. If the Spoiler moves in $\vec{X}$, then the Duplicator can follow the winning strategies for each $(G[X_i],S_i)$. On the other hand, if the Spoiler moves in $G-\vec{X}$, then the Duplicator can copy this move in $G'$.
\qed
\end{proof}
}

\lv{Next, we show that small representatives can be computed efficiently.}

\newcommand{\lemconstantrep}[0]{
\begin{lemma}
\label{lem:constantrep} Let $q$ be a non\hy negative integer
  constant. Let $G$ be a graph of rank-width at most $k$ and $S\subseteq V(G)$. Then there exists a function $f$ such that one can in time $f(k)\cdot |V(G)|^{\bigoh(1)}$ compute a graph $G'$ and a set $S'\subseteq V(G')$ such that
  $\Card{V(G')}$ is bounded by a constant and $\mathit{type}_q(G,S) =
  \mathit{type}_q(G',S')$.
\end{lemma}

\begin{proof}
  By Lemma~\ref{lem:typeformula} we can compute a formula $\Phi(Q)$ capturing the
  type $T$ of $(G,S)$ in time $f(k)\cdot |V(G)|^{\bigoh(1)}$. Given $\Phi(Q)$, a constant-size model $(G',S')$
  satisfying $\Phi(Q)$ can be computed as follows. We start enumerating all graphs (by brute force and in any order with a non-decreasing number of vertices),
  and check for each graph $G^*$ and every vertex-subset $S^*\subseteq V(G^*)$ whether $G^* \models \Phi(S^*)$. If this is
  the case, we stop and output $(G^*,S^*)$. Since $G \models \Phi(S)$ this procedure
  must terminate eventually. Fixing the order in which graphs are
  enumerated, the number of graphs we have to check depends only on $T$. By
  Fact~\ref{fact:representatives} the number of $q$\hy types is finite
  for each $q$, so we can think of the total number of checks and the size of each checked graph $G^*$ as bounded by a constant. Moreover the time spent on each check depends only on $T$ and the size of the graph $G^*$. Consequently, after we compute $\Phi(Q)$ it is possible to find a model for $\Phi(Q)$ in constant time.
  \qed
\end{proof}}

\lv{\lemconstantrep}

\lv{Finally, in Lemma~\ref{lem:similar} below we use Lemma \ref{lem:constantrep} to replace any \wsm{} by a small but ``equivalent'' modulator.}
\sv{The next lemma deals with actually computing small $q$-similar ``representatives'' for our \sm s.}

\lv{\begin{lemma}}
\sv{\begin{lemma}}
\label{lem:similar}
  Let $q$ be a non\hy negative integer constant and $\HH$ be a graph class. Then given a graph $G$ and a $k$-\wsm{} $\vec{X}=\{X_1,\dots X_k\}$ of $G$ into $\HH$, there exists a function $f$ such that one can in time $f(k)\cdot |V(G)|^{\bigoh(1)}$ compute a graph $G'$ with a $k$-\wsm{} $\vec{X'}=\{X'_1,\dots X'_k\}$ into $\HH$ such that $(G,\vec{X})$ and $(G',\vec{X'})$ are $q$-similar and for each $i \in [k]$ it holds that $|X_i'|$ is bounded by a constant.
\end{lemma}

\newcommand{\pflemsimilar}[0]{
\begin{proof}
  For $i\in [k]$, let $S_i\subseteq X_i$ be the frontier of split-module $X_i$, let $G_i=G[X_i]$ and let $G_0=G\setminus G[\vec{X}]$. We compute a graph $G_i'$ of constant size and a set $S_i'\subseteq V(G_i')$ with the same \MSO $q$-type as $(G_i,S_i)$. By Lemma~\ref{lem:constantrep}, this can be done in 
  time $f(k)\cdot |V(G)|^{\bigoh(1)}$ for some function $f$. Now let $G'$ be the graph obtained by the following procedure:
  
  \begin{enumerate}
  \item Perform a disjoint union of $G_0$ and $G'_i$ for each $i\in [k]$;
  \item If $k\geq 2$ then for each $1\leq i<j \leq k$ such that $S_i$ and $S_k$ are adjacent in $G$, we add edges between every $v\in S'_i$ and $w\in S'_j$.
  \item for every $v\in V(G_0)$ and $i\in [k]$ such that $S_i$ and $\{v\}$ are adjacent, we add edges between $v$ and every $w\in S'_i$.
  \end{enumerate}
  
   It is easy to verify that $(G,\vec{X})$ and $(G', \vec{X}')$, where $\vec{X}'=\{V(G_1'),\dots,V(G_k')\}$, are $q$\hy similar. \qed
\end{proof}
}

\lv{\pflemsimilar}

\sv{\begin{proof}[Sketch]
The idea here is to exploit the fact that each \sm{} has bounded rank-width. In particular, this allows us to determine the \MSO type of each $G[X_i]$ and its frontier $S_i$ in the specified time. The size of a minimum representative for each type does not depend on the actual size of $G$ or $k$.
\qed
\end{proof}
}

\begin{proof}[of Theorem~\ref{thm:main-use}]
Let $G$ be a graph, $k=\wsn^\HH(G)$ and $q$ be the nesting depth of quantifiers in $\phi$. By Theorem~\ref{thm:main-find} it is possible to find a $k$-\wsm{} to $\HH$ in time $f(k)\cdot |V|^{\bigoh(1)}$. We proceed by constructing $(G',\vec{X}')$ by Lemma~\ref{lem:similar}. Since each $X'_i\in\vec{X}'$ has size bounded by a constant and $|\vec{X}'|\leq k$, it follows that $\bigcup\vec{X}'$ is a modulator to the class of $\FF$-free graphs of cardinality $\bigoh(k)$. Hence $\MSOMC{\phi}$ can be decided in FPT time on $G'$. Finally, since $G$ and $G'$ are $q$-similar, it follows from Lemma~\ref{lem:partitiongame} that 
$G\models \phi$ if and only if $G'\models \phi$. \qed
\end{proof}

We conclude the section by showcasing an example application of Theorem~\ref{thm:main-use}.  \textsc{$c$-Coloring} asks whether the vertices of an input graph $G$ can be colored by $c$ colors so that each pair of neighbors have distinct colors. From the connection between \textsc{$c$-Coloring}, its generalization \textsc{List $c$-Coloring} and modulators~\cite[Theorem 3.3]{Cai03} and tractability results for \textsc{List-$c$-Coloring}~\cite[Page 5]{GolovachPaulusmaSong14}, we obtain the following.

\begin{corollary}
\label{cor:coloring}
\textsc{$c$-Coloring} parameterized by $\wsn^{P_5\text{-free}}$ is \emph{FPT} for each $c\in \Nat$.
\end{corollary}

\section{Conclusion}
\label{sec:hardness}

We have introduced a family of structural parameters which push the frontiers of fixed parameter tractability beyond rank-width and modulator size for a wide range of problems. In particular, the well-structure number can be computed efficiently (Theorem~\ref{thm:main-find}) and used to design FPT algorithms for \textsc{Minimum Vertex Cover}, \textsc{Maximum Clique} (Theorem~\ref{thm:problems}) as well as any problem which can be described by a sentence in MSO logic (Theorem~\ref{thm:main-use}).

In the wake of Theorem~\ref{thm:main-use} and the positive results for
the two problems in Section~\ref{sec:examples}, one would expect that
it should be possible to strengthen Theorem~\ref{thm:main-use} to also
cover LinEMSO
problems~\cite{CourcelleMakowskyRotics00,GanianHlineny10} (which
extend MSO Model Checking by allowing the minimization/maximization of
linear expressions over free set variables). Surprisingly, as our last
result we will show that this is in fact not possible if we wish to retain
the same conditions. For our hardness proof, it suffices to consider
a simplified variant of LinEMSO, defined below. Let $\varphi$ be an
MSO formula with one free set variable.
\begin{quote}
  $\MSOOPT{\leq}{\varphi}$\\
  \nopagebreak \emph{Instance}: A graph $G$ and an integer~$r\in
  \Nat$. \\ \nopagebreak \nopagebreak \emph{Question}: Is there a set
  $S \subseteq V(G)$ such that $G \models \varphi(S)$ and $\Card{S} \leq
  r$?
\end{quote}

\lv{
The following lemma will be useful later on. We say that $S\subseteq V(G)$ is a \emph{dominating set} if every vertex in $G$ either is in $S$ or has a neighbor in $S$.}

\newcommand{\lemdomfpt}[0]{
\begin{lemma}
\label{lem:domfpt}
The problem of finding a $p$-cardinality dominating set in a graph $G$ with a $k$-cardinality modulator $X\subseteq V(G)$ to the class of graphs of degree at most $3$ is FPT when parameterized by $p+k$.
\end{lemma}

\begin{proof}
Let $L=V(G)\setminus X$ and consider the following algorithm. We begin with $D=\emptyset$, and choose an arbitrary vertex $v\in L$ which is not yet dominated by $D$. We branch over the at most $k+4$ vertices $q$ in $\{v\}\cup N(v)$, and add $q$ to $D$. If $|D|=p$ and there still exists an undominated vertex in $G$, we discard the current branch; hence this procedure produces a total of at most $(k+4)^p$ branches.

Now consider a branch where $|D|<p$ but the only vertices left to dominate lie in $X$. For $a,b\in L$, we let $a\equiv b$ if and only if $N(a)\cap X=N(b)\cap X$. Notice that $\equiv$ has at most $2^k$ equivalence classes and that these may be computed in polynomial time. For each non-empty equivalence class of $\equiv$, we choose an arbitrary representative and construct the set $P$ of all such chosen representatives. We then branch over all subsets $Q$ of $P\cup X$ of cardinality at most $p-|D|$, and add $Q$ into $D$. Since $|P\cup X|\leq 2^k+k$, this can be done in time bounded by $\bigoh(2^{p\cdot k})$. Finally, we test whether this $D$ is a dominating set, and output the minimum dominating set obtained in this manner.

It is easily observed from the description that the running time is FPT. For correctness, from the final check it follows that any set outputed by the algorithm will be a dominating set. It remains to show that if there exists a dominating set of cardinality $p$, then the algorithm will find such a set. So, assume there exists a $p$-cardinality dominating set $D'$ in $G$. Consider the branch arising from the first branching rule obtained as follows. Let $v_1$ be the first undominated vertex in $L$ chosen by the algorithm, and consider the branch where an arbitrary $q\in D'\cap N(v_1)$ is placed into $D$. Hence, after the first branching, there is a branch where $D\subseteq D'$. Similarly, there exists a branch where $D\subseteq D'$ for each $v_i$ chosen in the $i$-th step of the first branching. If $D'=D$ after the first branching, then we are done; so, let $D'_1=D'\setminus D$ be non-empty. Let $D_1$ be obtained from $D'_1$ by replacing each $w\in D'_1$ by the representative of $[w]_\equiv$ chosen to lie in $P$. Since $D'$ dominates all vertices in $L$ and $D_1$ dominates the same vertices in $X$ as $D'_1$, it follows that $D^*=(D'\setminus D'_1)\cup D_1$ is also a dominating set of $G$. Furthermore, $|D^*|=|D'|$. However, since $D_1\subseteq P$ and $|D_1|\leq p-|D|$, there must exist a branch in the second branching which sets $Q=D_1$. Hence there exists a branch in the algorithm which obtains and outputs the set $D^*=D\cup D_1$.
\qed
\end{proof}}

\lv{\lemdomfpt}

\lv{\begin{theorem}}
\sv{\begin{theorem}}
\label{thm:opt-hard}
There exists an MSO formula $\varphi$ and a graph class $\HH$ characterized by a finite obstruction set such that $\MSOOPT{\leq}{\varphi}$ is \emph{FPT} parameterized by $\md^\HH$ but \emph{paraNP}-hard parameterized by $\wsn^\HH$.
\end{theorem}
\newcommand{\clopta}[0]{
\begin{claim}
$\MSOOPT{\leq}{\varphi}$ is FPT parameterized by the cardinality of a modulator to $\HH$.
\end{claim}

\begin{proof}[of Claim]
Let $(G=(V,E),r)$ be the input of $\MSOOPT{\leq}{\varphi}$ and $k$ be the cardinality of a modulator in $G$ to $\HH$. We begin by computing some modulator $X\subseteq V$ of cardinality $k$ in $G$ to $\HH$; this can be done in FPT time by a simple branching algorithm on any of the obstruction from $\FF$ located in $G$. Let $L=V\setminus X$.
Next, we compare $r$ and $k$, and if $r\geq k$ then we output YES. This is correct, since each $C_4$ in $G$ must intersect $X$ and hence setting $S=X$ satisfies $\varphi(S)$.

So, assume $r<k$. Then we check whether there exists a set $A$ of cardinality at most $r$ which intersects every $C_4$; this can be done in time $O^*(4^r)$ by a simple FPT branching algorithm. Next, we check whether there exists a dominating set $B$ in $G$ of cardinality at most $r$; this can also be done in FPT time by Lemma~\ref{lem:domfpt}.

Finally, if $A$ or $B$ exists, then we output YES and otherwise we output NO.
\lv{\hfill $\blacksquare$}
\sv{\qed}
\end{proof}}
\newcommand{\cloptb}[0]{
\begin{claim}
$\MSOOPT{\leq}{\varphi}$ is paraNP-hard parameterized by $\wsn^\HH(G)$.
\end{claim}

\begin{proof}[of Claim]
It is known that the \textsc{Dominating Set} problem, which takes as input a graph $G$ and an integer $j$ and asks to find a dominating set of size at most $j$, is NP-hard on $C_4$-free graphs of degree at most $3$~\cite{AlekseevKorobitsynLozin04}. We use this fact as the basis of our reduction. Let $(G,j)$ be a $C_4$-free instance of \textsc{Dominating Set} with degree at most $3$. Then we construct $G'$ from $G$ by adding $(|G|+2)$-many copies of $C_4$, a single vertex $q$ adjacent to every vertex of every such $C_4$, and a single vertex $q'$ adjacent to $q$ and an arbitrary vertex of $G$. It is easy to check that $\wsn^\HH(G')\leq 2$.

We claim that $(G,j)$ is a YES-instance of \textsc{Dominating Set} if and only if $(G',j+1)$ is a yes-instance of $\MSOOPT{\leq}{\varphi}$. Indeed, assume there exists a dominating set $D$ in $G$ of cardinality $j$. Then the set $D\cup \{q\}$ is a dominating set in $G'$, and hence satisfies $\varphi$. 

On the other hand, assume there exists a set $D'$ of cardinality at most $j+1$ which satisfies $\varphi$. If $j+1\geq |G|+2$ then clearly $(G,j)$ is a YES-instance of \textsc{Dominating Set}, so assume this is not the case. But then $D'$ cannot intersect every $C_4$, and hence $D'$ must be a dominating set of $G'$ of cardinality at most $j+1$. But this is only possible if $q\in D'$.  Furthermore, if $q'\in D'$, then replacing $q'$ with the neighbor of $q'$ in $G$ is also a dominating set of $G'$. Hence we may assume, w.l.o.g., that $D'\cap V(G)$ is a dominating set of cardinality at most $j$ in $V(G)$. Consequently, $(G,j)$ is a YES-instance of \textsc{Dominating Set} and the proof is complete.
\lv{\hfill $\blacksquare$}
\sv{\qed}
\end{proof}}
\lv{\begin{proof}}
To prove Theorem~\ref{thm:opt-hard}, we let $dom(S)$ express that $S$ is a dominating set in $G$, and let $cyc(S)$ express that $S$ intersects every $C_4$ (cycle of length $4$). Then we set $\varphi(S)=dom(S)\vee cyc(S)$
and let $\HH$ be the class of $C_4$-free graphs of degree at most $3$ (obtained by letting the obstrucion set $\FF$ contain $C_4$ and all $5$-vertex supergraphs of $K_{1,4}$).
\lv{\clopta}
\lv{\cloptb}
\lv{\qed}
\lv{\end{proof}}


We conclude with two remarks on Theorem~\ref{thm:opt-hard}. On one hand, the fixed parameter tractability of LinEMSO traditionally follows from the methods used for FPT MSO model checking, and in this respect the theorem is surprising. But on the other hand, our parameters are strictly more general than rank-width and hence one should expect that some results simply cannot be lifted to this more general setting.

\bibliographystyle{abbrv} \bibliography{literature}

\end{document}